\documentclass[11pt]{article}
\usepackage{amsmath,amssymb,amsfonts,amsthm,epsfig}
\usepackage[usenames,dvipsnames]{xcolor}
\usepackage{bm,xspace}
\usepackage{tcolorbox}
\usepackage{cancel}
\usepackage{fullpage}
\usepackage{liyang}
\usepackage{framed}
\usepackage{verbatim}
\usepackage{enumitem}
\usepackage{array}
\usepackage{multirow}
\usepackage{afterpage}
\usepackage{mathrsfs}
\usepackage{pifont} 
\usepackage{chngpage}
\usepackage[normalem]{ulem}
\usepackage{boxedminipage}
\usepackage{caption}
\usepackage{forest}

\usepackage{algorithm}
\usepackage{algorithmicx}
\usepackage{algpseudocode}

\usepackage{pgfplots}
\pgfplotsset{width=8cm,compat=newest}

\def\colorful{0}

\ifnum\colorful=1
\newcommand{\violet}[1]{{\color{violet}{#1}}}

\fi
\ifnum\colorful=0
\newcommand{\violet}[1]{{{#1}}}

\fi

\newcommand{\error}{\mathrm{error}}

\newcommand{\Sens}{\mathrm{Sens}}

\newcommand{\paren}[1]{\left({#1}\right)}
\DeclareMathOperator*{\argmax}{arg\,max}

\newlist{enumprop}{enumerate}{1} % set up a dedicated enumeration environment
\setlist[enumprop]{label=\arabic*.,ref=\theproposition.\arabic*}
\crefalias{enumpropi}{proposition}

\makeatletter
\newtheorem*{rep@theorem}{\rep@title}
\newcommand{\newreptheorem}[2]{
\newenvironment{rep#1}[1]{
 \def\rep@title{#2 \ref{##1}}
 \begin{rep@theorem}\itshape}
 {\end{rep@theorem}}}
\makeatother

\newreptheorem{theorem}{Theorem}

\begin{document}

%!TEX root = paper.tex

% \title{The Query Complexity of Certification \vspace{15pt}}

\title{The Query Complexity of Certification
%\footnote{Submission for non-archival track. To appear STOC 2022.} 
\vspace{15pt}}

\author{Guy Blanc \vspace{8pt} \\ \hspace{-5pt}{\sl Stanford}
\and Caleb Koch \vspace{8pt}\\ \hspace{-8pt} {\sl Stanford}
\and \hspace{0pt} Jane Lange \vspace{8pt} \\ \hspace{-4pt}  {\sl MIT}
\and Li-Yang Tan \vspace{8pt} \\ \hspace{-8pt} {\sl Stanford}}

\date{\vspace{15pt}\small{\today}}

\maketitle

\begin{abstract}
We study the problem of {\sl certification}: given queries to a function $f : \zo^n \to \zo$ with certificate complexity $\le k$ and an input $x^\star$, output a size-$k$ certificate for $f$'s value on~$x^\star$.  %This abstractly models a central problem in explainable machine learning, where we think of~$f$ as a blackbox model that we seek to explain the predictions~of. 

For monotone functions, a classic local search algorithm of Angluin accomplishes this task with~$n$ queries, which we show is optimal for local search algorithms.   Our main result is a new  algorithm for certifying monotone functions with $O(k^8 \log n)$ queries, which comes close to matching the information-theoretic lower bound of $\Omega(k \log n)$.  The design and analysis of our algorithm are based on a new connection to threshold phenomena in monotone functions.

We further prove exponential-in-$k$ lower bounds when $f$ is non-monotone, and when $f$ is monotone but the algorithm is only given random examples of $f$.  These lower bounds show that assumptions on the structure of $f$ and query access to it are both necessary for the polynomial dependence on $k$ that we achieve. 
\end{abstract} 

\thispagestyle{empty}
\newpage 
\setcounter{page}{1}
\section{Introduction}

Given a function $f : \zo^n\to\zo$ and an input $x^\star$, {\sl why} does $f$ output $f(x^\star)$ on  $x^\star$?  Among the many possibilities for what constitutes such an ``explanation", the notion of {\sl certificates} is perhaps  the simplest: a set $S \sse [n]$ of $x^\star$'s coordinates that determines $f$'s value on $x^\star$.  That is, $f(y) = f(x^\star)$ for all $y$ that agree with $x^\star$ on the coordinates in $S$.  

%In addition to their simplicity, there are several other appealing features of certificates as  explanations.  First, it is a representation-independent notion, and therefore can be used to explain a blackbox function $f$ for which we have no information about.  Furthermore, certificates capture representation-specific notions of explanations for many common types of representations.  For example, if $f$ a decision tree, a natural explanation for $f$'s value on $x^\star$ would be the unique root-to-leaf path that $x^\star$ follows; the set of coordinates queried along this path is a certificate for $x^\star$.  If $f$ is a linear threshold function $f(x) = \Ind[w\cdot x \ge \theta]$ and $x^\star$ is a satisfying assignment, a natural explanation would be a set of coordinates $S$ such that $x^\star_i = 1$ for all $i\in S$ and $\sum_{i\in S} w_i \ge \theta$; this set is again a certificate for~$x^\star$. 

It is natural to seek {\sl small} certificates, i.e.~succinct explanations: the smaller $S$ is, the more inputs it covers, and the more general it is as an explanation. This leads us to the following standard definition from complexity theory:

\begin{definition}[Certificate complexity]
\label{def:cert-complexity}
For a function $f : \zo^n \to \zo$ and an input $x^\star$, the {\sl complexity of certifying $f$'s value on $x^\star$} is the quantity: 
\[ C(f,x^\star) \coloneqq \min_{S\sse [n]}\big\{ |S| \colon \text{$f(y)= f(x^\star)$ for all $y$ s.t.~$y_S = x^\star_S$}\big\}.\] 
The {\sl certificate complexity of $f$} is the quantity $\ds C(f)\coloneqq \max_{x\in \zo^n} \{ C(f,x)\}$.
\end{definition}
We can now state the algorithmic problem that we study in this work, that of efficiently finding small certificates:  

\begingroup 
\addtolength\leftmargini{-0.1in}
\begin{quote} 
 {\bf Certification Problem:}  {\sl Given queries to a function $f : \zo^n\to\zo$ with certificate complexity $\le k$ and an input $x^\star$, output a size-$k$ certificate for $f$'s value on~$x^\star$.}
\end{quote} 
\endgroup

\paragraph{Motivation.} 
In addition to being a basic and natural problem, this is also an abstraction of a problem of interest in {\sl explainable machine learning}, where~$f$ represents a black box model that we seek to explain the predictions~of. 
 Modern machine learning algorithms, powered by large amounts of computational resources and trained on massive datasets, produce models that perform very well, but are so complicated that they are essentially inscrutable black boxes.  This is a concern as we increasingly delegate weighty decisions to these models.  The  field of explainable machine learning seeks to address this by developing techniques to  explain the predictions of these models~\cite{DK17,Lip18}.  
 
 There are numerous notions of ``explanations" in this literature~\cite{SK10,BSHKHM10,SVZ14,RSG16,KL17,LL17,STY17};  Ribero, Singh, and Guestrin~\cite{RSG18} were the first to propose certificates.  Their work introduced a relaxed ``approximate'' notion of certificates, where the set $S$ of coordinates {\sl mostly} determines $f$'s value rather than fully determines it, and ``mostly" is measured with respect to a distribution over inputs.  We discuss~\cite{RSG18}, this notion of ``approximate certificates", and corresponding approximate  certification algorithms in more detail in~\Cref{sec:prior}.

\subsection{Our results} 

%In this work we identify {\sl monotone} functions as broad and natural class of functions that admits extremely efficient certification algorithms.  

\subsubsection{Local search for monotone functions and its limitations} 
The certification problem can be viewed as the problem of efficiently finding an ``$f$-monochromatic" subcube in $\zo^n$ of codimension $\le k$ containing $x^\star$, where a subcube is $f$-monochromatic if $f$ takes the same value on all inputs in that subcube.  From this perspective, it is natural to proceed by {\sl local search}: first query $f$ on $x^\star$ and its immediate Hamming neighbors, and iteratively expand this neighborhood until it contains an $f$-monochromatic subcube of the desired size.

Indeed, a classic algorithm due to Angluin~\cite{Ang88} shows how such a local search can be carried out systematically for {\sl monotone} functions, and solves the certification problem with just~$n$ queries:\bigskip

%The starting point of our work is an observation that a simple algorithm due to Angluin~\cite{Ang88} solves the certification problem for monotone functions using $n$ queries: \bigskip

\noindent {\bf Angluin's algorithm:}  {\it Given queries to a monotone function $f:\zo^n \to \zo$ with certificate complexity $\le k$ and an input $x^\star$, Angluin's algorithm makes $n$ queries to $f$ and returns a size-$k$ certificate for $f$'s value on $x^\star$. }

\bigskip

Angluin's algorithm is a modification of a similar algorithm given by Valiant~\cite{Val84}.  

We begin by observing that Angluin's algorithm is optimal among local search algorithms.  We consider a local search algorithm to be any algorithm whose first query is~$x^\star$, and whose subsequent queries are Hamming neighbors of some input that has been previously queried.  In other words, at any point in the execution of a local search algorithm, the set of inputs that have been queried so far forms a connected subgraph of $\zo^n$ containing $x^\star$.  We show the following lower bound:

% \begin{claim}[Lower bound against local search algorithms] 
% \label{prop:local-lower-bound} 
% There is a constant $c > 0$ such that the following holds.  Any local search algorithm solving the certification problem for monotone functions $f : \zo^n\to \zo$ must have query complexity $\Omega(n)$, even if $f$ is promised to have certificate complexity $k=1$ and even if the algorithm is only required to return a size-$cn$ certificate w.h.p. 
% \end{claim} 

\begin{claim}[Lower bound against local search algorithms] 
\label{prop:local-lower-bound-v2} 
For any $\eps > 0$ the following holds.  Any local search algorithm solving the certification problem for monotone functions $f : \zo^n\to \zo$ must have query complexity $\Omega(\eps n)$, even if $f$ is promised to have certificate complexity $k=1$ and even if the algorithm is only required to return a size-$\Omega(\eps n)$ certificate with probability $\eps$. 
\end{claim}

\subsubsection{Near-optimal certification algorithm for monotone functions}

Our main result is an algorithm for certifying monotone functions that is substantially more efficient than Angluin's: 

\begin{theorem}[Efficient certification of monotone functions] 
\label{thm:main} 
Given queries to a monotone function $f:\zo^n \to \zo$ with certificate complexity $\le k$ and an input~$x^\star$, our algorithm makes $O(k^8\log n)$ queries to $f$ and w.h.p.~returns a size-$k$ certificate for $f$'s value on $x^\star$. 
\end{theorem} 

As one would expect given~\Cref{prop:local-lower-bound-v2}, our algorithm does not proceed by local search.  In fact, our algorithm takes the exact {\sl opposite} approach.  A local search algorithm for monotone functions starts with the trivial certificate $S = \{ i \in [n]\colon x^\star_i = f(x^\star_i)\}$ and trims it down in size by removing coordinates that are ``irrelevant to $S$".  Our algorithm proceeds the opposite way: we start with the empty set $S = \varnothing$ and add to it coordinates that we deem ``important".  We describe our approach in detail in~\Cref{sec:overview}.

We complement~\Cref{thm:main} with a lower bound showing that the query complexity of our algorithm is near optimal, even if the algorithm only has to return a certificate of size $\ell \gg k$: 

\begin{claim}[Lower bound for monotone functions]
\label{thm:approx-monotone-lb}
For any $c <1$ and any $k\le \ell \le n^c$, let $\mathcal{A}$ be an algorithm which, given query access to a monotone function $f : \zo^n\to \zo$ with certificate complexity $\le k$ and an input $x^\star$, returns a size-$\ell$ certificate for $f$'s value on $x^\star$ w.h.p.  The query complexity of $\mathcal{A}$ must be  $\Omega(k\log n)$. 
\end{claim}

\subsubsection{Algorithms and lower bounds for other settings} 

Finally, we study the extent to which the setting of~\Cref{thm:main} can be relaxed: what if $f$ is an arbitrary function, one that is not necessarily monotone?  What if the algorithm is only given uniformly-distributed random examples $(\bx,f(\bx))$ instead of query access to $f$?  We obtain fairly tight upper and lower bounds for both these settings. \Cref{table} summarizes these bounds and contrasts them with our results as described in the previous subsection:

\vspace{15pt}
\begin{table}[H]
  \captionsetup{width=.9\linewidth}
\begin{adjustwidth}{-4em}{}
\renewcommand{\arraystretch}{1.9}
\centering
%\begin{tabular}{|c|c|c|}
\begin{tabular}{|@{}c@{}|c|@{}c@{}|}
\hline
Algorithm is given:  & Upper bound  & Lower bound   \\ \hline \hline 
 \begin{tabular}{@{}c@{}}
Queries to monotone $f$,  \vspace{-10pt} \\
and proceeds by local search \end{tabular} & ~~Angluin's algorithm: $n$ queries~~& \begin{tabular}{c} 
\Cref{prop:local-lower-bound-v2}: $\Omega(n)$ queries 
\end{tabular} 
\\ 
\hline  Queries to monotone $f$ & ~~\Cref{thm:main}: $O(k^8\log n)$ queries~~& ~~\Cref{thm:approx-monotone-lb}: $\Omega(k\log n)$ queries~~\\ \hline
\begin{tabular}{c}
    Queries to  arbitrary $f$  \\ \hline 
     ~~Random examples of monotone $f$~~  
\end{tabular} & 
\Cref{claim:random-examples-ub}: $O(2^k k\log n)$ examples
& 
\begin{tabular}{c}
\Cref{claim:query-lb}: $\Omega(2^k + k\log n)$ queries  \\ \hline 
\Cref{claim:random-examples-lb}: $\Omega(2^k + k\log n)$ examples  
\end{tabular}
\\ 
\hline 
  \end{tabular} 
\end{adjustwidth}
\caption{Bounds on the query complexity of certification. }  
\label{table}
\end{table}

The exponential-in-$k$ lower bounds for these alternative settings (the last two rows of~\Cref{table}) show that some assumption on the structure of $f$, such as monotonicity, and query access to it are both necessary for the polynomial dependence on $k$ that we achieve in~\Cref{thm:main}. As in~\Cref{thm:approx-monotone-lb}, these lower bounds hold even if the algorithm is only required to return a size-$\ell$ certificate where $\ell$ can be significantly larger than~$k$; we defer the precise statements to the body of the paper.  

\subsection{Prior work on ``approximate" and exact certificates} 
\label{sec:prior} 

We begin by discussing two works from the explainable machine learning literature,~\cite{RSG18} and~\cite{BLT-Explanations}, that are direct precursors to ours. 

%The notion of certificates is central to computational complexity theory where it is the basis of the class NP.  The subfield of query complexity is devoted to the study of certificate complexity and related complexity measures of boolean functions. 
\vspace{-5pt} 
\paragraph{\cite{RSG18}.} Ribero, Singh, and Guestrin were the first to propose certificates as explanations for black box machine learning models. They introduced a relaxed notion of certificates that allows for errors\footnote{They termed such explanations {\sl anchors}, which has since become standard in the explainable machine learning literature. We stick with the term certificates in our description of their results.}:

\begin{definition}[Approximate certificates~\cite{RSG18}]
For a function $f : \zo^n \to \zo$, an input $x^\star$, a distribution $\mathcal{D}$ over $\zo^n$, and $\eps > 0$, we say that a set $S\sse [n]$ is an {\sl $\eps$-error certificate for $f$'s value on $x^\star$ with respect to $\mathcal{D}$} if  $\ds\Prx_{\by\sim \mathcal{D}}[\,f(\by)\ne f(x^\star)\mid \by_S = x^\star_S\,] \le \eps.$
\end{definition}

 \cite{RSG18}'s work was empirical in nature:  their paper demonstrated, through experiments and a user study, the effectiveness of succinct certificates as explanations.  %They attributed this effectiveness to the intuitive nature of certificates and their ease of understanding by human beings. 
 Their work also gave heuristics for finding succinct approximate certificates, but these heuristics do not come with provable performance guarantees.%\footnote{And indeed, as observed in~\cite{BLT-Explanations}, it is easy to construct examples where their heuristics fail badly.} 

\cite{RSG18}'s work has been influential in explainable machine learning. For more, see the discussion of their work in the book~\cite[Chapter~\S5.9]{molnar2020interpretable}, and the open source library~\cite{Alibi,Alibi-webpage} for implementation details of their heuristics.

\paragraph{\cite{BLT-Explanations}.} Motivated by~\cite{RSG18},~\cite{BLT-Explanations}  gave an algorithm for finding succinct approximate certificates that comes with performance guarantees with respect to the {\sl uniform distribution}:

\begin{theorem}[\cite{BLT-Explanations}'s approximate certification algorithm; informal]
\label{thm:BLT}
Let $\mathcal{U}$ denote the uniform distribution over $\zo^n$ and $\eps > 0$.  Given query access to~$f : \zo^n \to\zo$ with ``$\eps$-error certificate complexity" $\le k$ and an input~$x^\star$, \cite{BLT-Explanations}'s algorithm makes $\poly(k,1/\eps,n)$ queries to~$f$ and returns a set of coordinates $S(x^\star)$.   

With probability $\ge 1-\eps$ over $\bx^\star \sim \mathcal{U}$, the set $S(\bx^\star)$ is an $\eps$-error certificate for $f$'s value on~$\bx^\star$ with respect to $\mathcal{U}$ and   $|S(\bx^\star)|\le \poly(k,1/\eps)$.   
\end{theorem}

Comparing \Cref{thm:BLT} to our algorithm in~\Cref{thm:main}, we see that~\Cref{thm:BLT} applies to all functions whereas~\Cref{thm:main} only applies to monotone ones.  On the other hand, there are two sources of errors in~\Cref{thm:BLT}, neither of which are present in~\Cref{thm:main}: the guarantees of~\cite{BLT-Explanations}'s algorithm only hold for most $x^\star$ and not for all of them, and the certificates returned are $\eps$-error certificates and not actual certificates.   Even if one is willing to tolerate both sources of errors, the fact that they are measured with respect to the uniform distribution remains a significant shortcoming---this was identified in~\cite{BLT-Explanations} as the main limitation of their result.

A primary motivation for our work was to develop certification algorithms that, like~\cite{BLT-Explanations}'s, come with provable performance guarantees, but where these guarantees hold in the much more challenging {\sl errorless} setting.  

\paragraph{Other related work on finding certificates.} There has been significant work on finding prime implicants in the ML and AI literature (see e.g. \cite{I20, DH20, INS20, INM19} and the references therein), including for monotone functions \cite{SCD18,MGC+21}.  In our terminology, a prime implicant is a $1$-certificate which is minimal under set inclusion (relatedly a minimal $0$-certificate is a prime implicant for $\lnot f$). These algorithms for computing prime implicants all have worst-case query complexity and runtime that is at least linear in $n$. In contrast, our algorithm has only a logarithmic dependence on $n$ and always returns a prime implicant.

\section{Overview of our algorithm and its analysis} \label{sec:overview}

Before describing our algorithm, we first give an overview of Angluin's and~\cite{BLT-Explanations}'s algorithms, in tandem with a discussion of how these algorithms led to ours and how ours differs from them.  Throughout this section, let $f: \zo^n\to\zo$ be a monotone function and suppose without loss of generality that $f(x^\star) =1$ for the input~$x^\star$ that we seek to certify. 

\paragraph{Angluin's algorithm.}  By the monotonicity of $f$, the set $S_{x^\star} \coloneqq \{ i\in [n]\colon x^\star_i = 1\}$ is certainly a certificate for $f$'s value at~$x^\star$.  The assumption that $f$ has certificate complexity $\le k$ implies the existence of at least one subset $T\sse S_{x^\star}$ of size $\le k$  that remains a certificate for $f$'s value at $x^\star$. The goal of Angluin's algorithm is to find one of them.

\begin{definition}[Irrelevant coordinate of a certificate] 
\label{def:irrelevance} 
For a function $f$, an input $x^\star$, a certificate $S \sse [n]$ for $f$'s value at $x^\star$, and a coordinate $i\in S$, we say that $i$ is {\sl irrelevant to $S$} if $S\setminus \{ i\}$ remains a certificate for $f$'s value at $x^\star$, and otherwise say that it is {\sl relevant}. 
\end{definition}

 Angluin's algorithm starts with $S_{x^\star}$ and trims it down in size, removing irrelevant coordinates one by one, all the while maintaining the invariant that the current set remains a certificate.  A naive implementation of this plan results in a query complexity of $\Theta(|S_{x^\star}|^2)$.  A simple but key observation yields an improved query complexity of $O(|S_{x^\star}|) \le O(n)$: if $i$ is relevant for a certificate $S$, it remains relevant for any certificate $S'\sse S$.  Therefore, each coordinate $i\in S_{x^\star}$ is processed at exactly once throughout the entire execution of the algorithm.  (For completeness, we give a formal description of Angluin's algorithm and its analysis in~\Cref{appendix:angluin}.)  

\paragraph{\cite{BLT-Explanations}'s approximate certification algorithm.} \cite{BLT-Explanations}'s algorithm, as well as ours, takes an approach that is the  opposite of Angluin's, and indeed, the opposite of all local search algorithms.  Instead of starting with $S_{x^\star}$ and  removing irrelevant coordinates, we start with the empty set and add to it coordinates that we deem ``important".  The notion of {\sl influence} from the analysis of boolean functions provides a way to quantify the importance of coordinates:  

\begin{definition}[Influence] 
\label{def:influence} 
For a function $f: \zo^n\to\zo$ and a coordinate $i\in [n]$, the {\sl influence of $i$ on $f$} is the quantity $\Inf_i(f) \coloneqq \ds\Prx_{\mathrm{uniform}~\bx}[f(\bx)\ne f(\bx^{\oplus i})]$, where $\bx^{\oplus i}$ denotes $\bx$ with its $i$-th coordinate flipped. 
\end{definition} 

\cite{BLT-Explanations}'s algorithm is simple: using queries to $f$, determine the coordinate $i$ with (approximately) the largest influence\footnote{This is slightly imprecise, since~\cite{BLT-Explanations} actually uses a notion of ``noisy influence" which generalizes~\Cref{def:influence}.  We do not need this generalization in this work.} on $f$; restrict the $i$-th coordinate of $f$ according to $x^\star_i$ and recurse. \cite{BLT-Explanations} proved that for most $x^\star$'s, running this recursion to a certain depth suffices to guarantee a low-error certificate for $f$'s value on $x^\star$, where ``most" and ``low-error" are both with respect to the uniform distribution.

%Recalling~\Cref{thm:BLT},~\cite{BLT-Explanations} proved that with probability $1-\eps$ over $\bx^\star\sim\mathcal{U}$, running this recursion to depth $\poly(k,1/\eps)$ yields an $\eps$-error certificate for $f$'s value on $\bx^\star$ with respect to $\mathcal{U}$. 

\subsection{The three components of our algorithm} 

The difference between our setting  and~\cite{BLT-Explanations}'s is akin to the difference between  exact and uniform-distribution learning: exact learning is more challenging than distribution-independent learning, which is in turn more challenging than uniform-distribution learning.  \cite{BLT-Explanations}'s algorithm can be seen to fail badly in the setting of zero-error certificates: there are monotone functions $f$ with certificate complexity $k\ll n$ such that their recursion has to be run to the maximum depth of $n$ (corresponding to the trivial certificate $S=[n]$) in order to return a zero-error certificate.

Our algorithm is more involved than~\cite{BLT-Explanations}'s and has three main  components: 

\begin{enumerate} 
\item {\sl Finding a small certificate.} This component is independent of the input $x^\star$ that we seek to certify.  We design an algorithm that finds an {\sl arbitrary} $\poly(k)$-size certificate for a monotone~$f$---by arbitrary, we mean that this can be a certificate for $f$'s value on any input, not necessarily a specific one. In other words, this is a set $S\sse [n]$ and a bit $b\in \zo$ such that $f$ with all the coordinates $i\in S$ restricted to $b$ is a constant function.  

\item {\sl Finding a small certificate for $x^\star$.}  We then show how the algorithm above can be called $O(k)$ times to find a $\poly(k)$-size certificate for $f$'s value on $x^\star$.  The fact that $O(k)$ calls suffice follows from a basic result in query complexity, that every $1$-certificate and $0$-certificate of a function share at least one variable. (We defer the definitions of these terms to the body of the paper.)

\item {\sl Trimming the certificate.} Finally, we use Angluin's algorithm to trim the size of this certificate from $\poly(k)$ down to $\le k$. Crucially, we enter this trimming process with a certificate whose size is already bounded by $\le \poly(k)$, in contrast to Angluin's algorithm which starts with the certificate $S_{x^\star}$, the size of which can be as large as $n$.  The number of queries that we require for this step is therefore only $\le \poly(k)$, independent of $n$.
\end{enumerate} 

\subsubsection{Killing a monotone function}

We elaborate on the first component; the other two are fairly straightforward.  It will be useful for us to view this as the task of ``killing" a monotone function efficiently: using as few queries to $f$ as possible, find an assignment to a small set of coordinates that {\sl kills} $f$, meaning that the corresponding restriction of~$f$ is a constant function.

Our algorithm for this step is most easily understood from the perspective of {\sl threshold phenomena} in monotone functions---this connection is the key new ingredient in our work.  A wealth of techniques has been developed for the study of this topic, which is central to the theory of random graphs and percolation theory.  We will only need a few of the fundamentals. 

Every monotone function $f : \zo^n\to \zo$ can be associated with a function $\Phi_f : [0,1] \mapsto [0,1]$, 
\[ \Phi_f(p) \coloneqq \Ex_{\bx \sim \zo^n_{p}}[ f(\bx) ], \] 
where $\zo^n_{p}$ denotes the $p$-biased product distribution over $\zo^n$.  If $f$ is non-constant, this is a strictly increasing function of $p$, going from $0$ to $1$ as $p$ goes from $0$ to $1$. 

\begin{definition}[Critical probability] 
Let $f : \zo^n\to\zo$ be a non-constant monotone function.  The {\sl critical probability} of $f$ is the unique value $p(f) \in (0,1)$ for which $\Phi_f(p(f)) = \frac1{2}$.
\end{definition}

We use the critical probability of $f$ as a proxy for how close to constant it is, i.e.~how ``dead" the function is.  If $f$'s critical probability is $\ge \frac1{2}$, our algorithm kills it to the constant-$0$ function by driving its critical probability towards $1$; otherwise, we kill it to the constant-$1$ function by driving its critical probability towards $0$.   Our algorithm for doing so is similar in spirit to~\cite{BLT-Explanations}'s algorithm,  with the crucial difference being that ours ``continually adapts" to the critical probability of $f$ and its subfunctions: 
\begin{enumerate} 
\item Estimate the critical probability $p(f)$ of $f$. 
\item Determine the coordinate $i$ with approximately the largest {\sl $p(f)$-biased influence} on $f$.  The $p$-biased influence of a coordinate is the generalization of~\Cref{def:influence} to $p$-biased product distributions over $\zo^n$. 
\item Recurse on the subfunction $f_{x_i=b}$, the restriction of $f$ to $x_i = b$, where $b = 0$ if $p(f)\ge\frac1{2}$ and $b=1$ otherwise. 
\end{enumerate} 

Our analysis of this process relies on two basic results from the study of graph properties and percolation. We first use the O'Donnell--Saks--Schramm--Servedio inequality~\cite{OSSS05} to show that restricting~$f$ by the coordinate with the largest $p(f)$-biased influence changes its $p(f)$-biased expectation substantially: 
\[ \bigg| \Ex_{\text{$p(f)$-biased $\bx$}}[f(\bx)] - \Ex_{\text{$p(f)$-biased $\bx$}}[f_{x_i=b}(\bx)]\bigg|  \ge \Omega\left(\frac1{k^2}\right).\] 
We then show, via the Russo--Margulis lemma~\cite{Mar74,Rus78}, that the above implies that the critical probability of $f$ changes substantially: 
\begin{equation}  |p(f)-p(f_{x_i=b})| \ge \Omega\left(\frac1{k^3}\right). \label{eq:critical-probability-changes}
\end{equation} 
It follows that our algorithm kills $f$ within $O(k^3)$ recursive calls.  \Cref{fig:proof_slope_bound} on~\cpageref{fig:proof_slope_bound} illustrates our proof strategy. 

\paragraph{A slight optimization.}  The query complexity of this algorithm can be bounded by $O(k^8\log k\log n)$.  To shave off a factor of $\log k$, we consider an optimization where we estimate the critical probability of $f$ just once, at the very beginning of the algorithm, rather than in each recursive call.  Throughout the recursive process, we assume conservatively that each restriction only changes the critical probability by the minimum amount guaranteed by~\Cref{eq:critical-probability-changes}.  A simple adjustment of our analysis accounts for this modification (i.e.~for the possibility that the true critical probability drifts away from what we assume it to be as we recurse).

\section{Discussion and future work} 

Concrete directions for future work include closing the remaining gap between our upper and lower bounds of $O(k^8\log n)$ and $\Omega(k\log n)$, as well as  identifying other natural classes of functions that admit efficient certification algorithms.  

More broadly, a novel aspect of our techniques is the use of concepts and results from the study of threshold phenomena: $p$-biased analysis, the critical probability of monotone functions, the Russo--Margulis lemma, etc.  While the certification problem was the focus of this work, we speculate that there are further applications of this toolkit in learning theory, where monotonicity of the target function is a common assumption.  For example, while the variance of function is often used as progress measure in learning theory, our work suggests that for monotone target functions, its critical probability could be a more useful notion.  Can our idea of ``continually adapting" to the critical probability be used to design new learning algorithms?

Finally, circling back to the motivation for the certification problem, we mention that there is a growing flurry of work in explainable machine learning, the vast majority of which is empirical in nature; see slide 7 of~\cite{Kim-tutorial} for some staggering numbers.  Hallmarks of problems in this area---query access to a black box $f$ (``post-hoc explanations"); the focus on $f$'s values at and near a specific input $x^\star$ (``local explanations"); various notions of influence of variables (``feature attribution"); etc.---strongly suggest the potential for connections to areas of theoretical computer science such as query complexity, the analysis of boolean functions, learning theory, and sublinear algorithms.  Our work fleshes out a few of these connections, but we believe that there are more near at hand.

\section{Preliminaries}

We use {\bf boldface} often denote random variables (e.g. $\bx\sim\zo^n$) and we write ``w.h.p." to mean with probability $\ge 1-1/\poly(n)$.  We write $a = b\pm \eps$ as shorthand for $a \in [b-\eps,b+\eps]$.

\paragraph{Boolean function complexity.} 

In addition to certificate complexity (\Cref{def:cert-complexity}), we will need a few other standard notions and facts from boolean function complexity.  For an in-depth treatment (including proofs of the facts below), see~\cite{BdW02,Juk12}.

For a function $f : \zo^n \to \zo$ and an input $x\in \zo^n$, the {\sl sensitivity of $f$ at $x$} is the quantity 
\[ \Sens_f(x)=|\{i\in [n]:f(x)\neq f(x^{\oplus i})\}|, \] 
where $x^{\oplus i}$ denotes $f$ with its $i$-th coordinate flipped. 

\begin{proposition}[Sensitivity and certificate complexity]\label{prop:sens-at-most-cert} 
For all functions $f : \zo^n\to \zo$ and inputs $x\in\zo^n$, we have $\Sens_f(x) \le C_f(x)$. 
\end{proposition} 

For a function $f:\zo^n\to\zo$, we write $D(f)$ to denote its {\sl decision tree complexity}, the depth of the shallowest decision tree that computes $f$.

\begin{fact}[Decision tree complexity and certificate complexity]\label{fact:dt-at-most-cert2}
For all functions $f: \zo^n\to\zo$, we have $D(f) \le C(f)^2$. 
\end{fact} 

We also will occasionally distinguish between $0$-certificates and $1$-certificates.
\begin{definition}[$0,1$-certificate complexity]
For a function $f:\zo^n\to\zo$ and input $x\in\zo^n$, a certificate $S\subseteq [n]$ of $x$ is a $0$-certificate if $f(x)=0$ and $1$-certificate if $f(x)=1$. The $0$-certificate complexity and $1$-certificate complexity of $f$ are defined as 
$$
C_0(f)\coloneqq \max_{x\in f^{-1}(0)}\{C_f(x)\}~\text{ and }~C_1(f)\coloneqq \max_{x\in f^{-1}(1)}\{C_f(x)\}
$$
respectively.
\end{definition}

\paragraph{$p$-biased analysis.} We write $\{0,1\}_{p}^n$ to denote the $p$-biased product distribution on $n$ bit strings (that is, each bit is $1$ with probability $p$) and $\Pr_p$ to denote the $p$-biased probability measure on strings. When sampling from $\{0,1\}_{p}^n$, we will often just write the subscript $p$. In particular, $\E_p[f]$ denotes the expectation of $f$ with respect to $\bx \sim \{0,1\}_{p}^n$ and similarly $\Var_p[f]=\E_p[f^2]-\E_p[f]^2=\E_p[f](1-\E_p[f])$ is the $p$-biased variance of $f$.

We'll use two common notions of influence. 

\begin{definition}[$p$-biased flip influence; generalization of~\Cref{def:influence}]
\label{def:p-biased-influence}
Let $f : \zo^n\to\zo$ be a function, $p\in [0,1]$, and $i\in [n]$.  The {\sl $p$-biased flip influence of $i$ on $f$} is the quantity: 
\[ \Inf_{i,p}^{\oplus}[f]\coloneqq \Pr_p[f(\bx)\ne f(\bx^{\oplus i})]. \] 
\end{definition} 

\begin{definition}[$p$-biased rerandomized influence]
\label{def:p-biased-random-influence}
Let $f : \zo^n\to\zo$ be a function, $p\in [0,1]$, and $i\in [n]$. The $p$-biased rerandomized influence of $i$ on $f$ is the quantity:
$$
\Inf_{i,p}^{\sim}[f]\coloneqq 2\Pr_p[f(\bx)\ne f(\bx^{\sim i})]
$$
where $\bx^{\sim i}$ is the string $\bx$ with its $i$-th coordinate \textit{rerandomized} according to $\zo_p$.
\end{definition}

For each notion of influence, the \textit{total} influence is the sum of the influences of all the coordinates. We write $\Inf_p^{\oplus}[f]$ and $\Inf_p^{\sim}[f]$ for the total flip and rerandomized influence, respectively.

We record a few basic properties of $p$-biased influence and defer their proofs to~\Cref{appendix:p_analysis}:  

\begin{proposition}
\label{prop:basic_props}
For any boolean function $f:\{0,1\}^n\to\{0,1\}$ and $i\in [n]$, 
\begin{enumprop}
    \item $\Inf_p^{\oplus}[f]={\E}_p[\Sens_f(\bx)]$. \label{prop:basic_props:exp_sens}
    \item $\Inf_{i,p}^{\oplus}[f]=\Pr_p[f_{x_i=1}(\bx)\neq f_{x_i=0}(\bx)]$. \label{prop:basic_props:infi_pr}
    \item $\Inf_{i,p}^{\sim}[f]=4p(1-p)\Inf_{i,p}^{\oplus}[f]$. \label{prop:basic_props:infr=inff}
    \item $\Inf_p^{\sim}[f]\ge \Var_p[f]$. \label{prop:basic_props:inf_var}
\end{enumprop}
If $f$ is monotone, 
\begin{enumprop}[resume,itemsep=0.5em]
    \item $\E_p[f]=\E_p[f_{x_i=0}]+p\cdot \Inf_{i,p}^{\oplus}[f]=\E_p[f_{x_i=1}]-(1-p)\cdot \Inf_{i,p}^{\oplus}[f]$. \label{prop:basic_props:exp_monotone}
\end{enumprop}
\vspace{0.5em}
\end{proposition}

\section{First component of~\texorpdfstring{\Cref{thm:main}}{Theorem 1}: Finding an arbitrary certificate}

In this section, we show how to find an \textit{arbitrary} size-$\poly(k)$ certificate of a monotone function in $O(k^7\log n)$ queries where $k$ is the certificate complexity of the function. We first state the algorithm below then show each step can be implemented in a query efficient manner and with high probability of success. In particular, we'll give a $O(k^7\log k\log n)$ query upper bound and then we'll show how a simple modification of the algorithm can obtain a $O(k^7\log n)$ upper bound.

\begin{algorithm}
\begin{algorithmic}
\Require Query access to a monotone function $f:\{0,1\}^n\to \{0,1\}$ and parameter $k$.
\State Initialize $S\leftarrow \varnothing$
\While{$f$ is nonconstant}
\State Find an $\eps$-approximate critical probability $p$ of $f$, where $\eps = O(1/k^3)$
\State Estimate $\Inf_{i,p}^{\sim}[f]$ to additive accuracy $\pm O(1/k^2)$ for all $i$
\State Add coordinate $i$ to ${S}$ where $\Inf_{i,p}^{\sim}[f]$ is the largest influence estimate
\State $f\leftarrow f_{x_{i}=b}$ where $b=0$ if $p\ge 1/2$ and $1$ otherwise
\EndWhile
\State \Return the certificate $S$ \end{algorithmic}
\caption{Finding a certificate of a monotone function}
\label{alg:arbitrary_certificate}
\end{algorithm}

\begin{theorem}
\label{theorem:find_cert_algo_correctness}
Let $f:\zo^n\to\zo$ be a monotone function with $C(f)\le k$. There is an implementation of~\Cref{alg:arbitrary_certificate} that w.h.p.\ makes $O(k^7\log k\log n)$ queries to $f$ and returns a certificate of size $O(k^3)$.  
\end{theorem}

\subsection{Structural properties of \texorpdfstring{$\Phi_f$}{Phi}} 
As discussed in~\Cref{sec:overview}, the function $\Phi_f: [0,1]\to [0,1]$,
\[ \Phi_f(p) \coloneqq \E_p[f(\bx)] \] 
will be central to our analysis.  In this section we record and establish a few structural properties of $\Phi_f$ that will be useful for the proof of~\Cref{theorem:find_cert_algo_correctness}.  

The first is the Russo--Margulis lemma~\cite{Mar74,Rus78} which states that the derivative of $\Phi_f(p)$ is exactly the total flip influence of $f$ under the $p$-biased distribution. 

\begin{lemma}[Russo--Margulis]
    \label{lemma:russo_margulis}
    Let $f$ be a monotone function, then 
    $$
    \frac{d}{dp}\Phi_f(p)=\Inf_{p}^{\oplus}[f].
    $$
\end{lemma}

For a Fourier-analytic proof of the Russo--Margulis lemma, see~\cite{ODBook}. For the sake of completeness, we give a self-contained combinatorial proof in~\Cref{appendix:rm_lemma}.

We leverage three important corollaries of the Russo--Margulis lemma in our analysis.  Applying the lemma twice, to $\Phi_f(p)$ and $\Inf^{\oplus}_{i,p}[f]$, we can upper bound the Lipschitz constants of these quantities by $k$ when viewed as functions of $p$.    We then apply it again to lower bound the derivative of $\Phi_f(p)$ near the critical probability $p(f)$ of $f$, to show that that any $p$ for which $\Phi_f(p)$ is close to $1/2$ must be close to $p(f)$.

%The third corollary lower bounds total influence and uses Russo--Margulis in an analogous manner to show that any $p$ for which $\Phi_f(p)$ is close to $1/2$ must be close to the critical probability of $f$.

\begin{corollary}[Lipschitz constant of $\Phi_f$]
    \label{corollary:tangent_certificate}
    Let $f:\zo^n\to\zo$ be a monotone function with $C(f)\le k$, then for all $q \neq r$ we have
    $$
    \frac{\Phi_f(q)-\Phi_f(r)}{q-r}\le k.
    $$
\end{corollary}

\begin{proof}
    By the mean value theorem, the slope of the tangent line $(\Phi_f(q)-\Phi_f(r))/(q-r)$ is the derivative of $\Phi_f(p)$ at some point $\hat{p}$ in between $q$ and $r$. Applying the Russo--Margulis lemma, we have that 
    $$
    \frac{\Phi_f(q)-\Phi_f(r)}{q-r}=\left.\frac{d}{dp}\Phi_f(p)\right\rvert_{p=\hat p}=\Inf_{\hat p}^{\oplus}[f].
    $$
By~\Cref{prop:sens-at-most-cert,prop:basic_props:exp_sens}, 
    \[ \Inf_{\hat{p}}^{\oplus}[f] = \E_{\hat{p}}[\Sens_f(\bx)] \le \E_{\hat{p}}[C_f(\bx)] \le C(f)\] 
and the proof is complete. 
\end{proof}

\begin{corollary}[Lipschitz constant of $\Inf_{i,p}^{\oplus}$]
    \label{corollary:lipschitz_influence}
    Let $f:\zo^n\to\zo$ be a monotone function with $C(f)\le k$. Then for all $q\neq r$ and $i\in [n]$ we have
    $$
    \left\lvert\frac{\Inf_{i,q}^{\oplus}[f]-\Inf_{i,r}^{\oplus}[f]}{q-r}\right\rvert\le k.
    $$
\end{corollary}

\begin{proof}
    When $f$ is monotone, \Cref{prop:basic_props:infi_pr} can be written as $\Pr_p[f_{x_i=1}(\bx)\neq f_{x_i=0}(\bx)]=\Phi_{f_{x_i=1}}(p)-\Phi_{f_{x_i=0}}(p)$. Hence, 
    $$
    \frac{d}{dp}\Inf_{\hat{p}}^{\oplus}[f]=\frac{d}{dp}\left[\Phi_{f_{x_i=1}}(p)-\Phi_{f_{x_i=0}}(p)\right]=\Inf_p^{\oplus}[f_{x_i=1}]-\Inf_p^{\oplus}[f_{x_i=0}]
    $$
    by the Russo--Margulis lemma. Since $0\le\Inf_p^{\oplus}[f_{x_i=b}]\le C(f_{x_i=b})\le C(f)$ for $b\in\zo$, the result then follows from the application of the mean value theorem as in the proof of \Cref{corollary:tangent_certificate}.
\end{proof}

\begin{corollary}
    \label{corollary:p_progress_inverse}
    Let $f:\zo^n\to\zo$ be a monotone function and let $p\in [0,1]$ be any point satisfying $\Phi_f(p)=1/2\pm \eps$. Then
    $$
    p=p(f)\pm  \frac{4\eps}{1-4\eps^2}.
    $$
\end{corollary}

\begin{proof}
    Suppose without loss of generality that $p\le p(f)$ (the case where $p>p(f)$ is symmetric). Again applying the mean value theorem, there is some $\hat p\in [p,p(f)]$ satisfying $\Inf_{\hat p}^{\oplus}[f]=(\Phi_f(p(f))-\Phi_f(p))/(p(f)-p)$. Then, we have
    \begin{align*}
        \frac{\eps}{p(f)-p}&\ge \frac{\Phi_f(p(f))-\Phi_f(p)}{p(f)-p}\\
        &= \Inf_{\hat p}^{\oplus}[f]\ge \Var_{\hat p}[f] \tag{\Cref{prop:basic_props:inf_var}}\\
        &\ge \Var_p[f]=\Phi_f(p)(1-\Phi_f(p)) \tag{monotonicity}\\
        &\ge \paren{\frac{1}{2}+\eps}\paren{\frac{1}{2}-\eps}=\frac{1}{4}-\eps^2
    \end{align*}
    which gives the desired inequality.
\end{proof}

The next lemma quantifies the change in the critical probability of $f$ when we restrict one of its coordinates. In particular, we use the Lipschitz constant for $\Phi_f(p)$ to show this change is large when the restricted coordinate is influential. 

\begin{lemma} 
    \label{lem:p_progress} 
    Let $f:\zo^n\to\zo$ be a monotone function with $C(f) \le k$. Then for all $i\in [n]$, we have
    \begin{align*}
        p(f_{x_i=0})-p(f)&\ge \frac{p(f)\cdot \Inf_{i,p(f)}^{\oplus}[f]}{k}\\
        \shortintertext{and analogously,} p(f)-p(f_{x_i=1})&\ge \frac{(1-p(f))\cdot  \Inf_{i,p}^{\oplus}[f]}{k}.
    \end{align*}
\end{lemma} 
\begin{proof} 
    We prove the lower bound on $p(f_{x_i=0})-p(f)$. The proof for $p(f)-p(f_{x_i=1})$ is symmetric. First, rewriting \Cref{prop:basic_props:exp_monotone} in the $\Phi_f$ notation we have
    \begin{equation} 
        \Phi_{f_{x_i=0}}(p) =\Phi_f(p)-p\cdot \Inf_{i,p}^{\oplus}[f]. \label{eq:expectation-influence} 
    \end{equation}
    Therefore,
    \begin{align*}
        k&\ge \frac{\Phi_{f_{x_i=0}}(p(f_{x_i=0}))-\Phi_{f_{x_i=0}}(p(f))}{p(f_{x_i=0})-p(f)} \tag{\Cref{corollary:tangent_certificate}} \\
        &=\frac{\Phi_{f_{x_i=0}}(p(f_{x_i=0}))-\left(\Phi_f(p(f)) -p(f)\cdot \Inf_{i,p(f)}^{\oplus}[f]\right)}{p(f_{x_i=0})-p(f)} \tag{\Cref{eq:expectation-influence}} \\
        &= \frac{p(f)\cdot \Inf_{i,p(f)}^{\oplus}[f]}{p(f_{x_i=0})-p(f)}
    \end{align*}
    which completes the proof.
\end{proof} 

Finally, we need an inequality of O'Donnell, Saks, Schramm, and Servedio~\cite{OSSS05} which says that $f$ has an influential $p$-biased coordinate when the $p$-biased variance of $f$ is large. 

\begin{theorem}[OSSS inequality]
\label{thm:OSSS} 
For all functions $f : \zo^n\to\zo$ and $p\in [0,1]$, \[ \max_{i\in [n]} \big\{ \Inf_{i,p}^{\sim}[f]\big\} \ge \frac{\Var_p[f]}{D(f)},  \] 
where $D(f)$ denotes the decision tree complexity of $f$. 
\end{theorem}

\subsection{Algorithmic lemmas} 

We will need a few lemmas to bound the query complexity of~\Cref{alg:arbitrary_certificate}. First we show that we can find an approximation of the critical probability of $f$ by finding a value $p$ for which $\Phi_f(p)$ is close to $1/2$. Next we show that we can efficiently estimate \textit{rerandomized} influence to an additive accuracy. Finally, we show that if all of the influences are estimated under the $p$-biased distribution for $p$ close to $p(f)$, the critical probability of $f$, then the most influential coordinate under the $p$-biased distribution must also be influential under the $p(f)$-biased distribution. 

\begin{lemma}[Finding an approximate expectation of $f$]  
    \label{lem:estimate-crit-prob} 
    Given queries to a monotone $f:\zo^n\to\zo$ with $C(f)\le k$, for any $\eps >0$ we can find some $p\in [0,1]$ satisfying $\Phi_f(p)=1/2\pm \eps$ w.h.p.\ using $O(\log(k/\eps)\log(n)/\eps^2)$ many queries.  
\end{lemma}

\begin{proof} 
    Since $\Phi_f$ has Lipschitz constant $\le k$ (\Cref{corollary:tangent_certificate}), any value $\hat{p}$ that is within an additive $\pm \eps/3k$ of the true critical probability $p(f)$ of $f$ is an $\eps/3$-critical probability of $f$.  That is, 
    \[ \hat{p} = p(f)\pm \frac{\eps}{3k} \quad \Longrightarrow \quad \Phi_f(\hat{p}) = \frac1{2} \pm \frac{\eps}{3}. \] 
    
    We split the $[0,1]$ into $3k/\eps$ intervals each of length $\eps/3k$.  As observed above, the interval containing the critical probability will satisfy $\Phi_f(\hat{p}) = \frac1{2}\pm \frac{\eps}{3}$  for all $\hat{p}$ in that interval.  By the Chernoff bound, for any value $p\in [0,1]$ we can estimate $\Phi_f(p) = \E_p[f]$ to accuracy $\pm \eps/3$ and confidence $1-\delta$ using $O(\log(1/\delta)/\eps^2)$ queries.

    Performing binary search over the $3k/\eps$ intervals, with $O(\log(k/\eps))$ estimations of $\Phi_f(p)$ we are guaranteed to find a $\hat{p}$ such that our estimate of $\Phi_f(\hat{p})$ is $\frac1{2}\pm \frac{\eps}{3}\pm \frac{\eps}{3} = \frac1{2} \pm \frac{2\eps}{3}$; this implies that its true value is $\Phi_f(\hat{p}) = \frac1{2}\pm \frac{2\eps}{3} \pm \frac{\eps}{3} = \frac1{2} \pm \eps$, i.e.~$\hat{p}$ is indeed an $\eps$-approximate critical probability.   Choosing $\delta = 1/\poly(n)$ and noting that this is small enough to union bound over the $O(\log(k/\eps))$ many estimations (with much room to spare), we get that the overall query complexity is 
    \[ O(\log(k/\eps)) \cdot O(\log(n)/\eps^2) = O(\log(k/\eps)\log(n)/\eps^2). \qedhere \] 
\end{proof} 

\begin{lemma}[Finding an approximate critical probability]
    \label{lem:approx_crit_prob}
    Given queries to a monotone $f:\zo^n\to \zo$ with $C(f)\le k$ for any $0<\eps<1$, we can find $p\in [0,1]$ satisfying $p=p(f)\pm \eps$ w.h.p.\ using $O(\log (k/\eps)\log (n)/\eps^2)$ queries.  
\end{lemma}

\begin{proof}
    We show that any $p\in [0,1]$ satisfying $\Phi_f(p)=1/2\pm \eps/8$ satisfies the constraints of the lemma statement. The result then follows from \Cref{lem:estimate-crit-prob} which says that we can compute such a $p$ w.h.p.\ using $O(\log (k/\eps)\log n/\eps^2)$ queries.
    
    Let $p\in [0,1]$ satisfy $\E_p[f]=1/2\pm \eps/8$. Then we have
    \begin{align*}
        p&=p(f)\pm\frac{4(\eps/8)}{1-4(\eps/8)^2} \tag*{(\Cref{corollary:p_progress_inverse})}\\
        &=p(f)\pm\frac{\eps}{2-\eps^2/8}\\
        &=p(f)\pm \eps. \tag*{($\eps^2/8<1$)}
    \end{align*}
\end{proof}

\begin{lemma}[Estimating influences] 
\label{lem:estimate-influences}
Given queries to a monotone $f: \zo^n\to\zo$, some $p\in [0,1]$, and $\eps>0$, we can approximate $\Inf_{i,p}^{\sim}[f]$ to accuracy $\pm \eps$ for all $i\in [n]$ w.h.p.\ using $O(\log n/\eps^2)$ many queries. 
\end{lemma}

\begin{proof} 
Rewriting \Cref{prop:basic_props:exp_monotone} using $\Inf_{i,p}^{\sim}[f]=4p(1-p)\Inf_{i,p}^{\oplus}[f]$ we have
\begin{equation}
\Inf_{i,p}^{\sim}[f]=4(1-p)\paren{\E_p[f]-\E_p[f_{x_i=0}]}=4p\paren{\E_p[f_{x_i=1}]-\E_p[f]}.\label{eq:influence_exp_id}
\end{equation}
We show with a single random sample $\bS\subseteq \zo^n$ of size $O(\log n/\varepsilon^2)$ we can estimate $\Inf_{i,p}^{\sim}[f]$ to accuracy $\varepsilon$ for all $i\in[n]$ by estimating $\E_p[f]$ and either $\E_p[f_{x_i=1}]$ or $\E_p[f_{x_i=0}]$. We write $\overline{\E}_{\bS}[f]$ for the $p$-biased expectation of $f$ estimated from the set $\bS$. For each $i\in [n]$ and $b\in\zo$, we define $\bS_b=\{x^{-i}\in\zo^{n-1}:x\in\bS \text{ and }x_i=b\}$ where $x^{-i}$ denotes the string $x$ with the $i$th coordinate removed. Since $|\bS|=|\bS_1|+|\bS_0|$ we must have $|\bS_b|\ge |\bS|/2$ for some $b\in\zo$. We then estimate $\E_p[f_{x_i=b}]$ for this value of $b$ and use the appropriate identity from \cref{eq:influence_exp_id} to estimate the $i$th influence. Note that we can perform this estimate of $\E_p[f_{x_i=b}]$ because the strings in $\bS_b$ are distributed according to $\{0,1\}^{n-1}_p$ and we already know the values of $f_{x_i=b}$ for all strings in $\bS_b$ (since the query values of $f$ on $\bS$ are known). Thus by a Chernoff bound we can estimate both $\E_p[f_{x_i=b}]$ and $\E_p[f]$ to accuracy $\pm\varepsilon/8$ and confidence $1-\delta$ using $O(\log(1/\delta)/\varepsilon^2)$ random samples. These estimates then ensure that our estimate of $\Inf_{i,p}^{\sim}[f]$ has accuracy $\pm\varepsilon$. For example, if $b=0$, our estimates $\overline{\E}_{\bS}[f]$ and $\overline{\E}_{\bS_0}[f_{x_i=0}]$ satisfy
\begin{align*}
    \overline{\Inf}_{i,p}^{\sim}[f]&=4(1-p)\paren{\overline{\E}_{\bS}[f]-\overline{\E}_{\bS_0}[f_{x_i=0}]}\\
    &=4(1-p)\Bigl((\E_p[f]\pm \varepsilon/8)-(\E_p[f_{x_i=0}]\pm\varepsilon/8)\Bigr)\\
    &=\Inf_{i,p}^{\sim}[f]\pm(1-p)\varepsilon=\Inf_{i,p}^{\sim}[f]\pm\varepsilon
\end{align*}
where $\overline{\Inf}_{i,p}^{\sim}[f]$ denotes the influence estimate. We choose $\delta=1/\poly(n)$ small enough to union bound over all $i\in [n]$ which makes the total number of random samples/queries $O(\log n/\eps^2)$ as desired. 
\end{proof}

\begin{lemma}
\label{lem:inf_lower_bound}
Let $f:\zo^n\to\zo$ be a monotone function with $C(f)\le k$. Let $\overline{p}=p(f)\pm \eps$ for some $0<\eps<1/k^2$ and suppose $\overline{\Inf}_{i,\overline{p}}^{\sim}[f]=\Inf_{i,\overline{p}}^{\sim}[f]\pm k\eps$ for all $i\in [n]$. Then 
$$
\Inf_{i,p(f)}^{\oplus}[f]\ge \frac{1}{8k^2}-3k\eps
$$
where $i=\argmax_{i\in [n]} \overline{\Inf}_{i,\overline{p}}^{\sim}[f]$.
\end{lemma}

\begin{proof}
    Recall that $\Var_p[f]=\Phi_f(p)(1-\Phi_f(p))$ and for our estimate $\overline{p}=p(f)\pm \eps$ we have $\Phi_f(\overline{p})=1/2\pm k\eps$ since $\Phi_f$ has Lipschitz constant $\le k$ (\Cref{corollary:tangent_certificate}). Thus by monotonicity $\Var_{\overline{p}}[f]\ge (1/2-k\eps)(1/2+k\eps)\ge 1/8$ (using the assumption that $\eps<1/k^2$). The OSSS inequality, \Cref{thm:OSSS}, then states 
    $$
    \max_{i\in [n]} \big\{ \Inf_{i,\overline{p}}^{\sim}[f]\big\} \ge \frac{\Var_{\overline{p}}[f]}{D(f)}\ge \frac{1}{8D(f)}.
    $$
    Furthermore, we can lower bound $1/8D(f)\ge 1/8k^2$ using \Cref{fact:dt-at-most-cert2}. Since our estimate $\overline{\Inf}_{i,\overline{p}}^{\sim}[f]$ has accuracy $\pm k\eps$ the maximum influence estimate satisfies
    $$
    \max_i \overline{\Inf}_{i,\overline{p}}^{\sim}[f]\ge \max_i\Inf_{i,\overline{p}}^{\sim}[f]-k\eps\ge \frac{1}{8k^2}-k\eps. 
    $$
    Hence, the true influence at this maximal $i$ satisfies $\Inf_{i,\overline{p}}^{\sim}[f]\ge (1/8k^2-k\eps)-k\eps=1/8k^2-2k\eps$. Finally, to translate this bound to a lower bound on $\Inf_{i,p(f)}^{\oplus}[f]$ we switch to flip influence and apply our Lipschitz bound on $\Inf_{i,p}^{\oplus}$. In other words,
    \begin{align*}
        \Inf_{i,p(f)}^{\oplus}[f]&\ge \Inf_{i,\overline{p}}^{\oplus}[f]-|p(f)-\overline{p}|\cdot k \tag{\Cref{corollary:lipschitz_influence}}\\
        &\ge \Inf_{i,\overline{p}}^{\oplus}[f]-k\eps\\
        &\ge \Inf_{i,\overline{p}}^{\sim}[f]-k\eps\\
        &\ge \frac{1}{8k^2}-3k\eps.\qedhere
    \end{align*}
\end{proof}

\subsection{Proof of \texorpdfstring{\Cref{theorem:find_cert_algo_correctness}}{Theorem 3}}

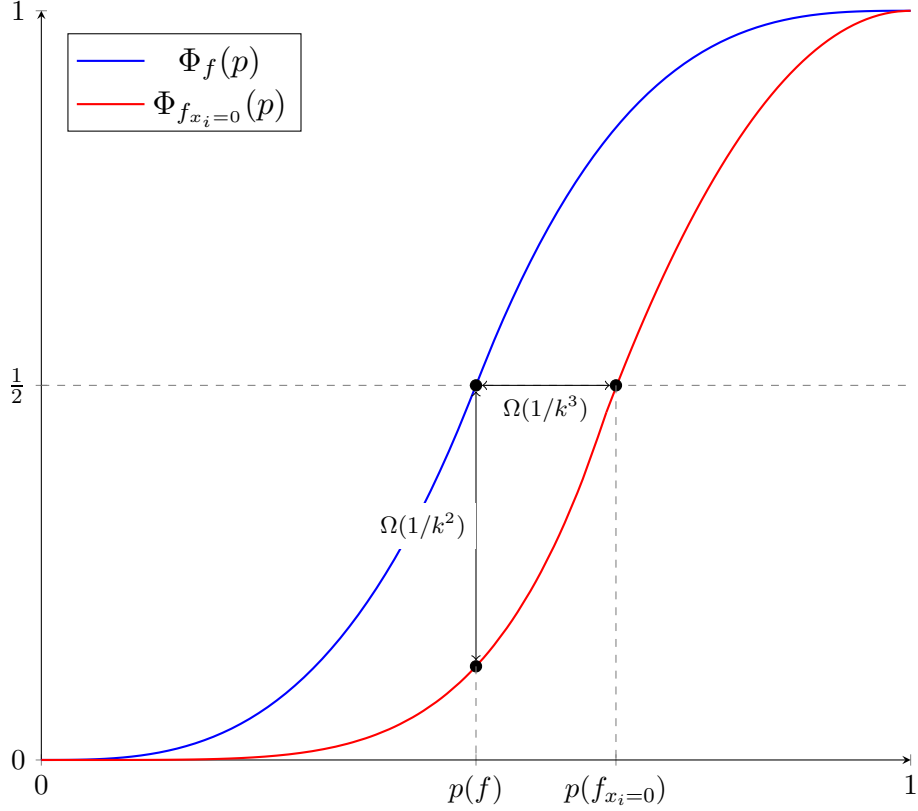
\begin{figure}
    \begin{center}
    % \resizebox{\linewidth}{!}{
    \begin{tikzpicture}
    \begin{axis}[
            scale only axis,
            width=0.7\linewidth,
            xtick={0,0.5, 0.661, 1},
            xticklabels={0,$p(f)$,$p(f_{x_i=0})$,1},
            ytick={0,0.5,1},
            yticklabels={0,$\frac{1}{2}$,1},
            xmin=0, xmax=1,
            ymin=0, ymax=1,
            ylabel = {},
            y label style = {rotate=-90},
            axis lines=left,
            clip=false,
            legend pos=north west,
            legend style={font=\small, nodes={scale=1.25, transform shape}},
            legend image post style={scale=1.5},
        ]
        % line y=1/2
        \addplot[gray, dashed, forget plot] coordinates {(0,0.5) (1,0.5)};

        % p-expectation of f
        \addplot[color=blue,smooth,thick,-,domain=0:0.5, forget plot] {4*x^3};
        \addplot[color=blue,smooth,thick,-,domain=0.5:1] {4*(x-1)^(3)+1};
        % \addlegendimage{line width=0.3mm,color=blue}
        \addlegendentry{ $\Phi_f(p)$}
        
        % p-expectation of f_0
        \addplot[color=red,smooth,thick,-,domain=0:0.65, forget plot] {4*x^5};
        \addplot[color=red,smooth,thick,-,domain=0.65:1] {-4.3746*(x-1)^2+1};
        \addlegendentry{ $\Phi_{f_{x_i=0}}(p)$}
        
        % derivative point and tangent line
        % \node[draw,circle,fill=black,inner sep=1.5pt] (derivative) at (0.5,0.221) {};
        % \addplot[mark=none, black] coordinates {(0.45,0.0996) (0.55,0.342)};
        
        % tangent line slope derivative 
        % \draw[color=black] (0.55,0.3) node[align=center] (label) [below right] { $\left.\frac{d\Phi_f(p)}{dp}\right\rvert_{p=\hat p}
        %     =\Inf_{\hat p}[f]$};
        
        % curved arrow
        % \draw[<-] (derivative) edge [out=0, in=180]  (label);
        
        % point E_{p(f)}[f0]
        \node[draw,circle,fill=black,inner sep=1.5pt] (epf0) at (0.5,0.125) {};
       
        % labels for p(f) and p(f0)
        % \draw[] (0.5,0.5) node [above,fill=white!30] {$p(f)$};
        % \draw[color=black] (0.661,0.5) node [above,fill=white] {$p(f_{x_i=0})$};
        
        % point p(f)
        \node[draw,circle,fill=black,inner sep=1.5pt] (p) at (0.5,0.5) {};
        
        % point p(f0)
        \node[draw,circle,fill=black,inner sep=1.5pt,] (p0) at (0.661,0.5) {};
       
        % arrow for p(f0)-p(f) distance
        % \draw [decorate,decoration={brace,raise=2pt}] (p) -- (p0) node [black,pos=0.5,yshift=0.5cm] { $\Omega(1/k^3)$};
        \draw [<->] (p) -- (p0) node[midway,below,fill=white!30] {\footnotesize $\Omega(1/k^3)$};
        
        % arrow for Epf[f0] to 1/2
        \draw[<->] (p) -- (epf0) node[midway,left,fill=white!30] {\footnotesize $\Omega(1/k^2)$};
        
        % tangent line connecting p0 and p1
        % \draw[black, thick, dotted] (epf0) -- (p0);
        
        \addplot[gray, dashed, forget plot] coordinates {(0.5,0.125) (0.5,0)};
        
        \addplot[gray, dashed, forget plot] coordinates {(0.661,0.5) (0.661,0)};
        
    \end{axis}
\end{tikzpicture}
% }
\end{center}
\caption{Illustration of the key atomic step in the proof of \Cref{theorem:find_cert_algo_correctness}. Let $p(f)$ denote the critical probability of $f$.  The OSSS inequality implies the existence of a coordinate $i\in [n]$ such that 
$\Phi_f(p(f)) - \Phi_{f_{x_i=0}}(p(f)) \ge  \Omega(k^{-2}),$
and we bound, using the Russo--Margulis lemma, the Lipschitz constant of $\Phi_{f_{x_i=0}}$ by $\le k$. We therefore conclude that the critical probabilities of $f$ and $f_{x_i=0}$ differ by $\Omega(k^{-3})$.}
\label{fig:proof_slope_bound}
\end{figure}

For our proof, we first show that accurate estimates of the critical probability of $f$ and the influences will ensure quick progress towards termination. Then we analyze the query complexity required to estimate these quantities to the specified accuracy with high confidence. This proof can be read in conjunction with \Cref{fig:proof_slope_bound} which illustrates the main idea.  

\paragraph{Proof of correctness.}
Our measure of progress is the critical probability of $f$. At a high level we show that if we find an $O(1/k^3)$-approximate critical probability and estimate influences to accuracy $O(1/k^2)$ at each step of the algorithm, then the critical probability of $f$ is guaranteed to increase or decrease by $\Omega(1/k^3)$. Since the function is constant when the critical probability is $0$ or $1$, we know that the algorithm must terminate after $O(k^3)$ steps. 

To be more specific, let $f$ be a nonconstant function obtained at some point in the algorithm with $C(f)\le k$. Let $0<\eps<1/k^2$ be arbitrary and let $\overline{p}=p(f)\pm \eps$ be an approximate critical probability and suppose each $\Inf_{i,p}^{\sim}[f]$ is estimated to accuracy $\pm k\eps$. Then, we can write
\begin{align*}
    p(f_{x_i=0})-p(f)&\ge \frac{p(f)\cdot \Inf_{i,p(f)}^{\oplus}[f]}{k}\tag{\Cref{lem:p_progress}}\\
    &\ge p(f)\cdot \paren{\frac{1}{8k^3}-3\eps}\tag{\Cref{lem:inf_lower_bound}}
\end{align*}
and likewise
$$
p(f)-p(f_{x_i=1})\ge (1-p(f))\cdot\paren{\frac{1}{8k^3}-3\eps}.
$$
In the final step of the algorithm's loop $f$ is restricted to $x_i=0$ if $\overline{p}\ge 1/2$ in which case we have $p(f)\ge 1/2-\eps$ and thus $p(f_{x_i=0})-p(f)\ge (1/2-\eps)(1/8k^3-3\eps)$. Note importantly that if $\overline{p}\ge 1/2$ the next estimate will also be greater than $1/2$ and so on, ensuring that the final certificate will be a $0$-certificate. We can then choose $\eps=O(1/k^3)$ small enough to ensure $p(f_{x_i=0})-p(f)\ge \Omega(1/k^3)$ and likewise in the case that $\overline{p}<1/2$. 

In both cases, one step of the main loop makes at least $\Omega(1/k^3)$ progress towards termination and so the loop iterates $O(k^3)$ times. Hence, the final certificate has at most $O(k^3)$ coordinates since each iteration of the loop adds one coordinate. 

\paragraph{Query complexity.}
\Cref{lem:estimate-crit-prob} shows we can compute a $O(1/k^3)$-approximate critical probability using $O(k^4\log k\log n)$ queries. Moreover, computing a $O(1/k^2)$-approximation of influence requires $O(k^4\log n)$ queries by \Cref{lem:estimate-influences}. Note also that we can test whether $f$ is constant with $\le 2$ queries using monotonicity ($f$ is constant if and only if $f(0^n)= f(1^n)$). Thus, one iteration of the main loop makes $O(k^4\log k\log n )$ queries to $f$. Since the main loop executes $O(k^3)$ times, the total number of queries is at most $O(k^7\log k\log n)$.

\subsection{Reducing the query complexity via fewer critical probability estimates}

We can reduce the query complexity of \Cref{alg:arbitrary_certificate} by a $\log k$ factor if we instead estimate the critical probability of $f$ once at the beginning of the algorithm then deterministically update it by the error term we calculated as $\eps$ in the proof above. At a high level, the idea is that the analysis for \Cref{theorem:find_cert_algo_correctness} shows that restricting $f$ by an influential coordinate will shift its critical probability by at least $\Omega(1/k^3)$. Hence, in the worst case, the algorithm makes the smallest amount of progress, approximately $1/k^3$, in each step. We can thus manually shift our critical probability estimate after each iteration by the minimal amount of progress we expect instead of using additional queries to $f$ to determine the new critical probability. 

In the lemma below we assume that the critical probability of $f$ is initially $\ge 1/2-\eps$, and hence $f$ is simplified by repeatedly restricting $0$-coordinates. The proof shows these restrictions force its critical probability to approach $1$. The alternate case where the initial critical probability is less than $1/2+\eps$ is analogous. In this case, one can show via symmetric arguments that the estimate $\overline{p}_t=1/2-(t-1)\eps$ satisfies $p(f_t)\le \overline{p}_t$ for all $t$.

%or larger than $1/2$.  and hence $f$ is simplified by repeatedly restricting $0$-coordinates. The proof shows these restrictions force its critical probability to approach $1$. The alternate case where the initial critical probability is less than $1/2$ is analogous.   

\begin{lemma}
    \label{lem:p_progress_deterministic}
    Fix an error term $0<\eps\le 1/40k^3$ and suppose $p(f)\ge 1/2-\eps$. Consider a variant of \Cref{alg:arbitrary_certificate} where at the $t^{\text{th}}$ step we estimate the critical probability as $\overline{p}_t=1/2+(t-1)\eps$ and we always set $f\leftarrow f_{x_i=0}$. Let $f_t:\zo^{n-t}\to\zo$ denote the function at the $t^{\text{th}}$ step. Then $p(f_t)\ge \overline{p}_t$ for all $t$ for which $f_t$ is nonconstant.
\end{lemma}

\begin{proof}
    The proof is by induction on $t$. The statement holds for $t=0$ by assumption. Otherwise assume that $p(f_t)\ge \overline{p}_t$. Then we show $p(f_{t+1})\ge \overline{p}_{t+1}$. If $p(f_t)> \overline{p}_{t+1}$ then there's nothing left to show since $\Phi_{f_t}(p(f_t))\le \Phi_{f_{t+1}}(p(f_{t+1}))$ always holds by \Cref{prop:basic_props:exp_monotone} and hence $p(f_t)\le p(f_{t+1})$. Otherwise, assume $p(f_t)\le \overline{p}_{t+1}$. In particular, $\overline{p}_t\le p(f_t)\le \overline{p}_{t+1}=\overline{p}_{t}+\eps$ which shows that $\overline{p}_{t}=p(f_t)\pm \eps$. The influence estimates have accuracy $\pm k\eps$ which then allows us to apply \Cref{lem:inf_lower_bound,lem:p_progress} as in the proof of \Cref{theorem:find_cert_algo_correctness} above, to conclude 
    $$
     p(f_{t+1})-p(f_t)\ge p(f_t)\cdot \paren{\frac{1}{8k^3}-3\eps}\ge\paren{\frac{1}{2}-\eps}\paren{\frac{1}{8k^3}-3\eps}.
    $$
    Choosing $\eps\le 1/40k^3$ then ensures $p(f_{t+1})-p(f_t)\ge \eps$ which completes the induction since $p(f_t)+\eps\ge \overline{p}_t+\eps=\overline{p}_{t+1}$.
\end{proof}

Equipped with \Cref{lem:p_progress_deterministic}, we can give a slight improvement on the query complexity of \Cref{theorem:find_cert_algo_correctness}. 

\begin{theorem}
\label{thm:arbitrary_cert_best_upperbound}
Given a monotone function $f:\zo^n\to\zo$ with $C(f)\le k$, there is an algorithm which w.h.p.\ returns a certificate of size $O(k^3)$ and makes $O(k^7\log n)$ queries to $f$. 
\end{theorem}

\begin{proof}
    We modify \Cref{alg:arbitrary_certificate} to estimate the critical probability of $f$ once at the start and then increment/decrement it by $\eps=1/(40k^3)$ after each iteration. Then the algorithm terminates after at most $O(k^3)$ iterations of the main loop by \Cref{lem:p_progress_deterministic}. We use \Cref{lem:approx_crit_prob} to estimate the critical probability of $f$ initially which requires $O(k^6\log k\log n)$ queries for our choice of $\eps$. Since this estimate $\overline{p}$ satisfies $\overline{p}=p(f)\pm \eps$ if $\overline{p}\ge 1/2$ then $p(f)\ge 1/2-\eps$ ensures the desired precondition for \Cref{lem:p_progress_deterministic} and otherwise $p(f)\le 1/2+\eps$ and the symmetric case applies.
    
    Each step of the algorithm's loop requires $O(k^4\log n)$ queries to estimate the influences to accuracy $k\eps=1/40k^2$ by \Cref{lem:estimate-influences}. Hence the algorithm makes $O(k^7\log n)$ queries overall.
\end{proof}

\section{Completing the proof of~\texorpdfstring{\Cref{thm:main}}{Theorem 1}}

In this section we show how to find a certificate for a given input using \Cref{alg:arbitrary_certificate} as a subroutine. The algorithm itself is fairly straightforward. For a monotone function $f$ and an input $x^\star$, we find an arbitrary certificate of $f$ using \Cref{alg:arbitrary_certificate} and then restrict $f$ on the coordinates in the certificate to the values specified by $x^\star$. Then we recurse on the subfunction and repeat until the function is constant. 

\begin{algorithm}
    \begin{algorithmic}[1]
    \Require A monotone function $f:\{0,1\}^n\to \{0,1\}$ and input $x^\star$.
    \State Initialize $S\leftarrow \varnothing$
    \While{$f$ is nonconstant}
    \State $s\leftarrow$ the output of \Cref{alg:arbitrary_certificate} on $f$
    \State $S\leftarrow S\cup s$ \Comment{Update certificate $S$ with coordinates from $s$}
    \State $f\leftarrow f_{x_i={x}^\star_i,i\in s}$ \Comment{restrict $f$ according to $s$}
    \EndWhile
    \State \Return the certificate $S$.
    \end{algorithmic}
    \caption{Finding a certificate for a given input}
    \label{alg:particular_certificate}
\end{algorithm}

We prove the following guarantee on \Cref{alg:particular_certificate}.

\begin{theorem}
    \label{theorem:particular_cert_algo_correctness}
    Let $f:\{0,1\}^n\to \{0,1\}$ be a monotone function with $C(f)\le k$, then \Cref{alg:particular_certificate} iterates $O(k)$ times and w.h.p. outputs a certificate of size $O(k^4)$. 
\end{theorem}

Combining this theorem with \Cref{thm:arbitrary_cert_best_upperbound}, we get that a certificate for an input to a monotone function can be found using at most $O(k^8\log n)$ queries to $f$.

\begin{corollary}
    \label{cor:particular_cert_upperbound}
    Let $f:\{0,1\}^n\to \{0,1\}$ be a monotone function with $C(f)\le k$. Then, a certificate of size $O(k^4)$ can be computed w.h.p. for any input $x^\star$ using $O(k^8\log n)$ queries to $f$.
\end{corollary}

The progress measure in our analysis of \Cref{alg:particular_certificate} is $C_0(f)+C_1(f)$, the sum of the $0$-certificate complexity and $1$-certificate complexity of $f$. In particular, each iteration of the main loop is guaranteed to decrease this quantity by at least $1$ which gives an upper bound on $2C(f)$ on the total number of iterations. For the proof, we use the fact that, for any Boolean function, every $0$-certificate must intersect every $1$-certificate (since otherwise there would be one input string having both a $0$-certificate and a $1$-certificate). 

\begin{fact}
    \label{fact:cert_intersections}
    Let $S_0$ be a $0$-certificate for a Boolean function $f:\zo^n\to\zo$ and let $S_1$ be a $1$-certificate. Then $S_0\cap S_1\neq\varnothing$.
\end{fact}

\begin{proof}[Proof of \Cref{theorem:particular_cert_algo_correctness}]
    Let $f$ be a nonconstant function during the execution of the algorithm. We'll show that $C_0(f)+C_1(f)$ decreases by at least $1$ after each iteration of the main loop. Let $s$ denote the certificate that \Cref{alg:arbitrary_certificate} returns and suppose without loss of generality that $s$ is a $1$-certificate (the argument is symmetric for a $0$-certificate). Then we'll show that $C_0(f_s)\le C_0(f)-1$ where $f_s$ is the restriction according to $s$ and ${x}^\star$: $f_s=f_{x_i={x}^\star_i,i\in s}$. Consider any $x\in f_s^{-1}(0)$. Let $x'\in\zo^n$ be the string formed by inserting $x^\star|_s$ into the string $x$ so that $f(x')=f_s(x)$ and $x'|_s=x^\star|_s$. Let $s_0$ be a $0$-certificate of $f$ on $x'$ with $|s_0|\le C_0(f)$. Then $s_0\setminus s$ is a $0$-certificate of $f_s$ on $x$. We can bound the size of this $0$-certificate:
    \begin{align*}
        |s_0\setminus s|&\le |s_0|-1 \tag{$s_0\cap s\neq\varnothing$ by \Cref{fact:cert_intersections}}\\
        &\le C_0(f)-1.
    \end{align*}
     Since $x$ is any arbitrary $0$-input to $f_s$, we have that $C_0(f_s)\le C_0(f)-1$ as desired.

    Since $f$ must be constant when either $C_0(f)$ or $C_1(f)$ is $0$, the algorithm must terminate after at most $C_0(f)+C_1(f)\le 2C(f)$ iterations. Each iteration adds at most $O(k^3)$ coordinates to the certificate $S$ and hence $|S|$ is $O(k^4)$ at the end of the algorithm.  
\end{proof}

\subsection{Trimming the certificate using Angluin's algorithm}

\Cref{alg:particular_certificate} returns a certificate of size $O(k^4)$. In this section, we show how to reduce that certificate to size $\le k$ using $O(k^4)$ additional queries.

\begin{claim}
    \label{claim:trim_cert}
    Let $S$ be a certificate for an input $x^\star$ of a monotone function $f:\zo^n\to\zo$. If $|S|>C(f)$ then a certificate $S'\subseteq S$ with $|S'|\le C(f)$ can be computed from $S$ using $O(|S|)$ queries to $f$.
\end{claim}

The proof of this claim is implicit in \cite[Theorem 1]{Ang88}. We give a self-contained exposition of the proof adapted to our setting in \Cref{appendix:angluin}.

We apply \Cref{claim:trim_cert} as a postprocessing step after executing \Cref{alg:particular_certificate}. Since this postprocessing step only requires an additional $O(k^4)$ queries to $f$ the overall number of queries is still upper bounded by $O(k^8\log n)$, the query bound on \Cref{alg:particular_certificate}. Thus, the combination of \Cref{cor:particular_cert_upperbound} with \Cref{claim:trim_cert} establishes \Cref{thm:main}.

\section{Lower bounds: Proofs of~\texorpdfstring{\Cref{prop:local-lower-bound-v2}}{Claim 1.1} and~\texorpdfstring{\Cref{thm:approx-monotone-lb}}{Claim 1.2}}

Our lower bounds in this section and the next  will rely on the easy direction of Yao's lemma:

\violet{ \begin{lemma}[\cite{Yao77}]
    \label{lem:Yao}
    For any $q \in \N$, let $\mathcal{R}_q$ and $\mathcal{D}_q$ be the set of all $q$-query randomized and deterministic algorithms respectively, and let $I$ be the set of all possible pairs $f: \zo^n \to \zo$ and $x^\star \in \zo^n$ (i.e.~instances of the certification problem). 
    
    For any distribution $\mu$ supported on $I$,
    \begin{align*}
        \min_{R \in \mathcal{R}_q} \max_{(f,x^\star) \in I}[\error_R(f,x^\star)] \geq \min_{D \in \mathcal{D}_q} \Ex_{(\boldf,\bx^\star) \sim \mu}[\error_D(\boldf,\bx^\star)]
    \end{align*}
    where $\error_R(f,x^\star)$ is the probability that $R$ does not successfully return a certificate for $f$'s value on $x^\star$, and $\error_D(f,x^\star) = \Ind[\text{$D$ does not successfully return a certificate for $f$'s value on $x^\star$}]$.
\end{lemma}}

% \gray{ \begin{lemma}[\cite{Yao77}]
%     %\label{lem:Yao}
%     For any $q \in \N$, let $\mathcal{R}_q$ and $\mathcal{D}_q$ be the set of all $q$-query randomized and deterministic algorithms respectively, and let $S$ be the set all of possible pairs $f: \zo^n \to \zo$ and $x \in \zo^n$. For any distribution $\mu$ supported on $S$,
%     \begin{align*}
%         \min_{\mathcal{A} \in \mathcal{R}_q} \max_{(f,x) \in S}[\error_\mathcal{A}(f,x)] \geq \min_{\mathcal{A} \in \mathcal{D}_q} \Ex_{(\boldf,\bx) \sim \mu}[\error_\mathcal{A}(\boldf,\bx)]
%     \end{align*}
%     where $\error_\mathcal{A}(f,x)$ is the probability that $\mathcal{A}$ does not successfully return a certificate for $f$'s value on $x$.
% \end{lemma}
% }

\subsection{Proof of~\texorpdfstring{\Cref{prop:local-lower-bound-v2}}{Claim 1.1}}

 \Cref{prop:local-lower-bound-v2} is a special case of the following claim: 
\begin{claim}
    \label{lem:lb-local}
    Let $n, q, \ell \in \N$ and $\mathcal{A}$ be a $q$-query randomized local search algorithm. There is a monotone $f: \zo^n \to \zo$ with $C(f) = 1$ and input $x^\star \in \zo^n$ on which $\mathcal{A}$ successfully returns a size-$\ell$ certificate for $x$ with probability $\le (\ell+q-1)/n$.
\end{claim}
We use Yao's lemma with the distribution $\mu$ where: 
\begin{enumerate}
    \item $\bx$ is a constant, supported entirely on $x^\star = [1,\ldots, 1]$, and 
    \item $\boldf$ is a random dictator: we select $\bi \in [n]$ uniformly at random and set $\boldf(x) = x_{\bi}$. 
\end{enumerate}

We will assume that $\mathcal{A}$ is deterministic and prove that the probability, over the randomness of $\boldf$, that $\mathcal{A}$ successfully finds a size-$\ell$ certificate $\boldf$'s value on $x^\star$ is at most  $(\ell+q-1)/n)$.

\violet{ 
\begin{proposition}
    \label{prop:local-mostly-one}
    Let $\mathcal{A}$ be any deterministic $q$-query local search algorithm.  For any $f: \zo^n \to \zo$, let $x^{(1)},\ldots,x^{(q)}$ be $\mathcal{A}$'s queries when it is asked to certify $f$'s value on $x^\star = [1,\ldots,1]$. The number of coordinates $i$ on which $x^{(j)}_i = 0$ for some $j \in [q]$ is at most $q - 1$.%\lnote{Tiny bit stronger than what's stated below.} \gnote{Isn't this weaker? I think we need that the number of coordinates $i$ for which $x^{(j)}_i = 0$ for any $j \in [q]$ is at most $q-1$.}
\end{proposition}
}

% \gray{ 
% \begin{proposition}
%     %\label{prop:local-mostly-one}
%     For any $n,q \in \N$,  $f: \zo^n \to \zo$ and any deterministic local search algorithm $\mathcal{A}$, let $x^{(1)}, \ldots, x^{(q)}$ be the sequence of queries that $\mathcal{A}$ makes when determining a certificate for $x^\star = [1,\ldots, 1]$ on $f$. The number of coordinates $i$ on which $x^{(j)}_i = 0$ for some $j \in [q]$ is at most $q - 1$.
% \end{proposition}
% } 
\begin{proof}
    By induction on $j$. For $j = 1$, a local search algorithm's first query must be $x^{(1)} = x^\star = [1,\ldots,1]$ which has no coordinates set to $0$.   For $j > 1$, we know that $x^{(j)}$ is Hamming adjacent to some $x^{(j')}$ where $j' < j$. Thus, $x^{(j)}$ can have at most one coordinate $i$ on which $x^{(j)}_i = 0$ but $x^{(j')}_i=1$. The desired result holds by induction.
\end{proof}

%The remainder of this proof is quite similar to that of \Cref{lem:non-monotone-lb}. 

\violet{ 
\begin{proposition}
    \label{prop:local-all-one-queries}
    Let $\mathcal{A}$ be any deterministic $q$-query local search algorithm and $\boldf:\zo^n \to \zo$ be a uniformly random dictator.  The probability, over the randomness of $\boldf$, that $\boldf$'s value is $0$ on least one of $\mathcal{A}$'s queries is at most $(q-1)/n$.
\end{proposition}}

% \gray{ 
% \begin{proposition}
% %    \label{prop:local-all-one-queries}
%     For any $n,q \in \N$, let $\boldf:\zo^n \to \zo$ be a uniformly random dictator and $\mathcal{A}$ be any deterministic local search algorithm that makes $q$ queries. The probability that the output of at least one of $\mathcal{A}$'s queries is $0$ is at most $\frac{q-1}{n}$.
% \end{proposition}} 
\begin{proof}
 For each $j \in [q]$, let $x^{(j)}$ be $\mathcal{A}$'s $j^{\text{th}}$ query when $\boldf$'s value on its first $j-1$ queries are all $1$. Note that $\boldf$'s value is $0$ on at least one of $\mathcal{A}$'s queries iff $\boldf(x^{(j)}) = 0$ for some $j\in [q]$. Hence
    \begin{align*}
        \Prx_{\boldf}\big[\text{$\boldf$'s value is $0$ on at least one of $\mathcal{A}$'s queries}\big] &= \Prx_{\boldf}\big[\boldf(x^{(j)}) = 0 \text{ for some } j \in [q]\big] \\
        &= \Prx_{\bi \in [n]}\big[x^{(j)}_{\bi} = 0 \text{ for some $j \in [q]$}\big]  \tag{Definition of $\boldf$} \\
        &\leq \frac{q-1}{n}. \tag{\Cref{prop:local-mostly-one}}
    \end{align*}
\end{proof}
We upper bound the probability any \violet{set $S$ of size $\ell$} is a certificate for $\boldf$'s value on $x^\star = [1,\ldots,1]$.

\violet{ 
\begin{proposition}
    \label{prop:local-overlap}
    Fix any set $S \sse [n]$ of size $\ell$.  The probability, over the randomness of $\boldf$, that $S$ is a certificate for $\boldf$'s value on $x^\star = [1,\ldots,1]$ is at most $\ell/n$.
\end{proposition}
}
% \gray{ 
% \begin{proposition}
%     %\label{prop:local-overlap}
%     For any $\ell, n  \in \N$, let $\boldf:\zo^n \to \zo$ be a uniformly random dictator and $S \subseteq [n]$ be any fixed set of $\ell$ coordinates. The probability that $S$ is a certificate for $x^\star$ is at most $\frac{\ell}{n}$.
% \end{proposition}} 

\begin{proof}
 Recall that $\boldf(x) = x_{\bi}$ for uniformly random $\bi \in [n]$.  Therefore $S$ is a certificate for $\boldf$'s value on $x^\star$ iff $\bi \in S$, which happens with probability $|S|/n = \ell/n$.
\end{proof}

With~\Cref{prop:local-all-one-queries,prop:local-overlap}, we can now complete the proof of~\Cref{lem:lb-local}: 

\violet{ 
\begin{proof}[Proof of~\Cref{lem:lb-local}]
As $\mathcal{A}$ is a deterministic algorithm, when $\boldf$'s values on $\mathcal{A}$'s queries are all $1$, there is a single set of coordinates $S$ output by $\mathcal{A}$. Then,

%Let $S$ be the set of coordinates output by $\mathcal{A}$ when $\boldf$'s values on its queries are all $1$. Then, 
\begin{align*}
    &\Prx_{\boldf}\big[\mathcal{A} \text{ returns a size-$\ell$ certificate for $\boldf$'s value on $x^\star$}\big] \\
    =& \Prx_{\boldf}\big[\mathcal{A}\text{ returns a size-$\ell$ certificate for $\boldf$'s value on $x^\star$}\ \&\ \text{$\boldf$'s values on all queries are $1$}\big] \\
    &\quad+\Prx_{\boldf}\big[\mathcal{A} \text{ returns a size-$\ell$ certificate for $\boldf$'s value on $x^\star$}\ \& \  \text{$\boldf$'s value on some query is $0$}\big]\\
    \leq& \Prx_{\boldf}\big[\text{$S$ is a certificate for $\boldf$'s value on $x^\star$}\big] + \Prx_{\boldf}\big[\text{$\boldf$'s value is $0$ on at least one of $\mathcal{A}$'s queries}\big] \\
    \leq& \frac{\ell}{n} + \frac{q-1}{n}. \tag{\Cref{prop:local-all-one-queries,prop:local-overlap}}
\end{align*}
\end{proof} } 
% \gray{ \begin{proof}[Proof of~\Cref{lem:lb-local}]
% Let $S$ be the set of coordinates output by $\mathcal{A}$ when every query it makes outputs $1$. Then,
% \begin{align*}
%     \Pr_{\bf}\big[\mathcal{A}& \text{ finds a correct certificate for }x^\star \big] \\
%     &= \Pr_{\bf}\big[\mathcal{A}\text{ finds a correct certificate for }x^\star , \text{ all queries output $1$}\big] \\
%     &\quad+\Pr_{\bf}\big[\mathcal{A} \text{ finds a correct certificate for }x^\star , \text{ some query outputs $0$}\big]\\
%     &\leq \Pr_{\bf}\big[S \text{ is a certificate for }x^\star\big] + \Pr_{\bf}\big[\text{Some query outputs $0$}\big] \\
%     &\leq \frac{\ell}{n} + \frac{q-1}{n}. \tag{\Cref{prop:local-all-one-queries,prop:local-overlap}}
% \end{align*}
% \end{proof} }

\subsection{Proof of~\texorpdfstring{\Cref{thm:approx-monotone-lb}}{Claim 1.2}}

 The proof  is simple and is essentially an instantiation of the following elementary fact:  if a problem $P$ has $\ge M$ possible outputs, and the input to $P$ can be accessed only via queries with binary answers, then $\log M$ is a lower bound on the query complexity of solving~$P$.  In our context of certification, since there are ${n \choose k}$ many sets of size $k$, this fact suggests that if every such set is a possible certificate, then $\log({n \choose k}) \approx k\log n$ would be a lower bound on query complexity.  Indeed this is what we show, and the argument extends easily to certification algorithms that are allowed to return a certificate of size $\ell \ge k$:

% \bigskip 

% This section formalizes an error-robust, randomized version of the following fact, which is typically discussed in the context of comparison sorting.

% \begin{fact}[Folklore]
% If a problem $P$ has $M$ possible outputs and the input to the problem
% can be accessed only via binary queries, then $\log M$ is a query lower bound
% for $P$.
% \end{fact}

% Since there are $\binom{n}{k}$ possible monotone 1-certificates of size $k$, this fact would suggest that if every possible 1-certificate is a possible output, then $\log(\binom{n}{k}) \approx k \log n$ is a query lower bound for the monotone certificate search problem. Indeed, that is what we show in this section. More formally, we show the following: 

\begin{claim}
    \label{lem:monotone-lb}
   Let $k, \ell, n, q \in \N$  and $\mathcal{A}$ be a $q$-query randomized algorithm. There is some monotone function $f: \zo^n \to \zo$ with $C(f) \le  k$ and input $x^\star \in \zo^n$ on which $\mathcal{A}$ successfully returns a size-$\ell$ certificate for $x^\star$ with probability at most $2^q \cdot \binom{\ell}{k} / \binom{n}{k} \leq 2^q \cdot \left(\frac{\ell e}{n}\right)^k$.
\end{claim}

\Cref{thm:approx-monotone-lb} follows as an immediate consequence of~\Cref{lem:monotone-lb}: if $k\le \ell \le n^c$ for any $c < 1$, then $q\ge  \Omega(k\log n)$ queries are necessary even to succeed with probability $0.1$.

% \gray{ 
% \lnote{I think we can skip this and just write the above?} Solving for $q$ gives the following lower bound of $\Omega(k \log n)$ queries to achieve nontrivial accuracy for monotone functions.
% \begin{corollary}
%     \label{cor:monotone-lb}
%     For any $k, n, q \in \N$ and $\eps > 0$ where $k \le \ell$ and $k, \ell = O(n^c)$ for some constant $c < 1$, let $\mathcal{A}$ be a possibly randomized algorithm making at most $q$ queries. If $\mathcal{A}$ successfully finds a size-$k$ certificate for any $x \in \{0,1\}^n$ and any monotone $f: \zo^n \to \zo$ with $C(f) = k$ with probability at least $\eps$, then $q = \Omega(\log \eps + k \log n)$.
% \end{corollary}}

\begin{proof}
We will once again use Yao's lemma.  Consider the distribution $\mu$ where:  

\begin{enumerate}
    \item $\bx$ is constant, supported entirely on $x^\star = [1,\ldots, 1]$, and
    \item $\boldf$ is drawn uniformly at random from the set of monotone conjunctions of $k$ variables.
\end{enumerate}

We observe that if $f$ is the monotone conjunction of the variables some set $T$, then a set $S$ certifies $f$'s value on $x^\star$ iff $S \supseteq T$.  Therefore, for any fixed set $S$ of size at most $\ell$, 
\begin{align*}
    \Prx_{\boldf}\big[\text{$S$ certifies $\boldf$'s value on $x^\star$}\big] &=\Prx_{\boldf}\big[\text{$\boldf$ is a conjunction of $k$ variables within $S$}\big]  \\
    &= \frac{{|S| \choose k}}{\binom{n}{k}} \le \frac{{\ell \choose k}}{\binom{n}{k}}.
\end{align*} 
Since any deterministic $q$-query algorithm $\mathcal{A}$ can take on at most $2^q$ many output values, we have by a union bound that
\begin{align*}
    \Prx_{\boldf}\big[ \text{$\mathcal{A}$ finds a size-$\ell$ certificate for $\boldf$'s value on $x^\star$}\big] \le 2^q \cdot \frac{\binom{\ell}{k}}{\binom{n}{k}} \leq 2^q \cdot \frac{(\ell e/k)^k}{(n/k)^k} = 2^q \cdot \left(\frac{\ell e}{n}\right)^k. 
\end{align*}
\Cref{lem:monotone-lb} follows from the above and an application of Yao's lemma. \end{proof} 
% Therefore, if $\mathcal{A}$ is any deterministic $q$-query algorithm, 
% Let $\mathcal{A}'$ be a \emph{deterministic} $q$-query algorithm that returns a size-$k$ certificate. We will first show that the probability over $\mu$ that the output of $\mathcal{A}'$ correctly certifies $(\boldf, \bx)$ is at most $2^q / \binom{n}{k}$. \\

% Each set of size $\ell$ is a certificate of exactly $\binom{\ell}{k}$ distinct $k$-conjunctions. Since $\mathcal{A}'$ is a deterministic algorithm, it can have at most one output for each sequence of query answers in $\zo^q$. Then the total number of $k$-conjunctions that are certified by any output of $A'$ is at most $2^q \cdot \binom{\ell}{k}$. Since $\boldf$ is drawn uniformly from the set of $k$-conjunctions, the probability that any output of $\mathcal{A}'$ is a certificate for $\boldf$ is at most 
% \[2^q \cdot \frac{\binom{\ell}{k}}{\binom{n}{k}}\]
% This quantity upper bounds the probability over $\mu$ that $\mathcal{A}$ outputs a correct certificate for $\boldf$. \Cref{lem:monotone-lb} follows from the application of Yao's lemma.
%\end{proof}

\section{Algorithms and lower bounds for other settings}

\subsection{An algorithm for certifying arbitrary functions with random examples}

\begin{claim}
\label{claim:random-examples-ub} 
    For any $k, m, n \in \N$, there is an algorithm which, given access to uniform random samples $(\bx,f(\bx))$ of a function $f: \zo^n \to \zo$ with certificate complexity $\leq k$, an input $x^\star \in \zo^n$, \violet{and $f$'s value on $x^\star$}, uses $m$ random samples and returns a size-$k$ certificate for $f$'s value on $x^\star$ with probability at least
    \begin{align*}
        1 - (1 - 2^{-k})^m\cdot \binom{n}{k}.
    \end{align*}
\end{claim} 
In particular, the algorithm succeeds with high probability if $m = \Theta(2^k  k  \log n)$.

Our proof of \Cref{claim:random-examples-ub} uses the following easy fact: 
\begin{proposition}
    \label{prop:high-var}
    For every non-constant $f: \zo^n \to \zo$ with certificate complexity $\leq k$ and every $b\in \zo$, 
\[         \Prx_{\bx \sim \zo^n}[f(\bx) = b] \geq 2^{-k}.\] \end{proposition}
\begin{proof}
    Without loss of generality, we only prove that the probability $f(\bx) = 1$ is at least $2^{-k}$. As~$f$ is non-constant, there is some input $y$ on which $f(y) = 1$. Since $f$ has certificate complexity $\leq k$, there is some set $S$ of size $\le k$ where $f(x) = 1$ whenever $x_S = y_S$. Finally,
    \begin{equation*}
        \Prx_{\bx \sim \zo^n}[f(\bx) = 1]  \geq  \Prx_{\bx \sim \zo^n}[\bx_S = y_S] \geq 2^{-k}. \qedhere 
    \end{equation*}
\end{proof}

\begin{proof}[Proof of \Cref{claim:random-examples-ub}]
We say that a set $S \subseteq [n]$ is {\sl eliminated} by a sample $(x,f(x))$ if $x_S = x^\star_S$ and $f(x) \neq f(x^\star)$.    The algorithm is simple:  it iterates over all $\binom{n}{k}$ candidate size-$k$ certificates (i.e.~all size-$k$ sets), keeping only those not eliminated by any of the $m$ sample points, and returns an arbitrary one. %Then, it arbitrarily selects one of the certificates that was not eliminated. 
Any actual certificate for $f$'s value on $x^\star$ will not be eliminated by the above procedure. Therefore, if all non-certificates are eliminated, the output of this algorithm will be correct. 
    
    Fix any size-$k$ set $S$ that is {\sl not} a certificate for $f$'s value on $x^\star$, \violet{ and consider $f_{x^{\star}_S}$, the subfunction of $f$ obtained by restricting the coordinates in $S$ according to $x^\star$.  Since $f$ has certificate $\le k$, all its subfunctions, including $f_{x^{\star}_S}$, also have certificate complexity $\le k$.  Furthermore, since $S$ is not a certificate for $f$'s value on $x^\star$, we have that $f_{x^{\star}_S}$ is non-constant.  Hence, by~\Cref{prop:high-var},
    \begin{align*}
        \Prx_{\bx \sim \zo^n}[f_{x^{\star}_S}(\bx) \neq f(x^\star)] \geq 2^{-k}.
    \end{align*}
    Therefore, the probability a random sample $(\bx,f(\bx))$ eliminates $S$ is at least $2^{-k}$. Since the samples are independent,  the probability $S$ is not eliminated after $m$ samples is at most $(1 - 2^{-k})^m$. Union bounding over all $\binom{n}{k}$ possible non-certificates $S$ of size $k$ gives the desired result.
    }
    % \gray{ Let $f_{S, x^\star}(x)$ be equal to $f(y_{x})$ where
    % \begin{align*}
    %     (y_x)_i \coloneqq \begin{cases}
    %     x^\star_i & \text{if $i \in S$} \\
    %     x_i &\text{otherwise}
    %     \end{cases}
    % \end{align*}
    % As $f$ has certificate complexity $\leq k$, all of its restrictions, including $f_{S, x^\star}$, also have certificate complexity $\leq k$. Furthermore, since $S$ is an invalid certificate for $x^\star$, $f_{S, x^\star}$ is not a constant function. By \Cref{prop:high-var},
    % \begin{align*}
    %     \Prx_{\bx \sim \zo^n}[f_{S, x^\star}(\bx) \neq f(x^\star)] \geq 2^{-k}.
    % \end{align*}
    % Therefore, the probability a random sample eliminates $S$ is at least $2^{-k}$. The samples are independent, so the probability $S$ remains after $m$ samples is at most $(1 - 2^{-k})^m$. Union bounding over all $\binom{n}{k}$ possible certificates $S$ gives the desired result.} 
\end{proof}

\subsection{Lower bound on the query complexity of certifying an arbitrary function} 

%This subsection is devoted to proving the following.

\violet{ 
\begin{claim}
    \label{claim:query-lb}
  Let  $k, n, q, \ell \in \N$ and $\mathcal{A}$ be a $q$-query randomized algorithm. There is some $f: \zo^n \to \zo$ with $C(f) = k$ and input $x^\star \in \zo^n$ on which $\mathcal{A}$ successfully returns a size-$\ell$ certificate for $x^\star$ with probability at most $q\cdot 2^{-k} + (k\ell)/n$.
\end{claim}

\Cref{claim:query-lb} implies that as long as $k\le \ell$ satisfy $k\ell\le 0.01n$, then $q\ge  \Omega(2^k)$ queries are necessary even to succeed with probability $0.1$.  Combining this with the $q\ge \Omega(k\log n)$ lower bound we showed in~\Cref{thm:approx-monotone-lb} yields the $q\ge \Omega(2^k + k\log n)$ lower bound stated in~\Cref{table}.  
}

% \gray{ 
% \lnote{Combined both of the following into the above} \begin{claim}
% %\label{claim:query-lb} 
%      For any $k, n, q, \ell \in \N$ and $\eps > 0$ where $k \ell = O(\eps  n)$, let $\mathcal{A}$ be a $q$-query randomized algorithm. If $\mathcal{A}$ successfully finds a size-$\ell$ certificate for any $x \in \{0,1\}^n$ and $f: \zo^n \to \zo$ satisfying $C(f) = k$ with probability at least $\eps$, then
%     \begin{align*}
%         q = \Omega(\eps  2^k).
%     \end{align*}
% \end{claim} 

% \Cref{claim:query-lb} is an easy corollary of the following Lemma.
% \begin{lemma}
%     \label{lem:non-monotone-lb}
%     Let $k, n, q, \ell \in \N$, and $\mathcal{A}$ be a possibly randomized algorithm making at most $q$ queries. There is some $f: \zo^n \to \zo$ with $C(f) = k$ and input $x \in \zo^n$ on which $\mathcal{A}$ successfully returns a size-$\ell$ certificate for $x$ with probability at most $\frac{q}{2^k} + \frac{k \ell}{n}$.
%  \end{lemma}} 

We apply Yao's lemma with the distribution $\mu$ where: 
\begin{enumerate}
    \item $\bx$ is constant, supported entirely on $x^\star = [1,\ldots, 1]$,
    \item $\boldf$ is the indicator function of a uniformly random subcube \violet{of codimension $k$}. More formally, we select $k$ uniformly random unique coordinates $\bi_1, \bi_2, \ldots, \bi_k \in [n]$ and $k$ uniform random bits $\bb_1, \bb_2, \ldots, \bb_k \sim \zo$, and let: 
    \begin{align*}
            \boldf(x) = \begin{cases}
            1 & \text{if }x_{\bi_j} = \bb_j \text{ for all $j \in [k]$} \\
            0 & \text{otherwise}.
            \end{cases}
    \end{align*}
\end{enumerate}

By Yao's lemma, in order to prove \Cref{claim:query-lb}, we need only show that every $q$-query deterministic strategy successfully finds a size-$\ell$ certificate for $x^\star$ with probability at most $\frac{q}{2^k} + \frac{k  \ell}{n}$ (over the randomness of $\boldf$). The proof of~\Cref{claim:query-lb} is similiar in spirit to~\Cref{lem:lb-local}, and will follow from \Cref{prop:all-zero-queries,prop:overlap}:

\violet{
\begin{proposition} 
\label{prop:all-zero-queries}
Let $\mathcal{A}$ be a $q$-query deterministic algorithm.  The probability, over the randomness of $\boldf$, that $\boldf$'s value is $1$ on at least one of $\mathcal{A}$'s queries is at most $q\cdot 2^{-k}$. 
\end{proposition} 
}
% \gray{ 
% \begin{proposition}
%     %\label{prop:all-zero-queries}
%     For any $k, q \in \N$, let $\boldf$ be the indicator function of a uniformly random size-$k$ subcube and $\mathcal{A}$ be any deterministic $q$-query algorithm. The probability that the output of at least one of $\mathcal{A}$'s queries is $1$ is at most $\frac{q}{2^k}$.
% \end{proposition}} 
\begin{proof}
    Since $\mathcal{A}$ is a deterministic algorithm, the queries it makes are a deterministic function of the previous query outputs. For each $j \in [q]$, let $x^{(j)}$ be $\mathcal{A}$'s $j^{\text{th}}$ query when $\boldf$'s value on its first $j-1$ queries are all $0$. Note that $\boldf$'s value is $1$ on at least one of $\mathcal{A}$'s queries iff there is some $j \in [q]$ for which $\boldf(x^{(j)}) = 1$. Hence
    \begin{align*}
        \Prx_{\boldf}\big[\text{$\boldf$'s value is $1$ on at least one of $\mathcal{A}$'s queries}\big] &= \Prx_{\boldf}\big[\boldf(x^{(j)}) = 1 \text{ for some } j \in [q]\big] \\
        &\leq \sum_{j \in [q]} \Prx_{\boldf}\big[\boldf(x^{(j)}) = 1\big]  \tag{Union bound} \\
        &= \frac{q}{2^k}. \tag*{(Definition of $\boldf$)}
    \end{align*}
\end{proof}

\violet{ 
\begin{proposition}
    \label{prop:overlap}
    Fix a set $S\sse [n]$ of size $\ell$.  The probability, over the randomness of $\boldf$, that $S$ is a certificate for $\boldf$'s value on $x^\star = [1,\ldots,1]$ is at most $(k\ell)/n$. 
\end{proposition}

} 

% \gray{ \begin{proposition}
%     %\label{prop:overlap}
%     For any $k,\ell, n  \in \N$, let $\boldf$ be the indicator function of a uniformly random size-$k$ subcube and let $S \subseteq [n]$ be any fixed set of $\ell$ coordinates. The probability that $S$ is a certificate for $x^\star$ is at most $\frac{k  \ell}{n}$.
% \end{proposition}} 
\begin{proof}
Recall that $\boldf$ is a function of $k$ random coordinates $\bi_1,\ldots,\bi_k\sim [n]$. 
In order for $S$ to be a certificate for $\boldf$'s value on $x^\star$, it has to contain at least one $\bi_j$.   Hence,
    \begin{align*}
        \Prx_{\boldf}\big[S \text{ is a certificate for \violet{$\boldf$'s value on} }x^\star\big] &\leq \Prx_{\boldf}\big[\bi_j \in S \text{ for some $j \in [k]$}\big] \\
        &\leq \sum_{j \in [k]} \Prx_{\boldf}\big[\bi_j \in S \big] \tag{Union bound} \\
        &\leq k \cdot \frac{\ell}{n}.\qedhere 
    \end{align*}
\end{proof}

\violet{ 
\begin{proof}[Proof of~\Cref{claim:query-lb}] 
Let $S$ be the set of coordinates output by $\mathcal{A}$ when $\boldf$'s values on its queries are all $0$.  Then, 
\begin{align*}
    &\Prx_{\boldf}\big[\mathcal{A} \text{ returns a size-$\ell$ certificate for $\boldf$'s value on $x^\star$}\big] \\
    =& \Prx_{\boldf}\big[\mathcal{A}\text{ returns a size-$\ell$ certificate for $\boldf$'s value on $x^\star$}\ \&\ \text{$\boldf$'s values on all queries are $0$}\big] \\
    &\quad+\Prx_{\boldf}\big[\mathcal{A} \text{ returns a size-$\ell$ certificate for $\boldf$'s value on $x^\star$}\ \& \  \text{$\boldf$'s value on some query is $1$}\big]\\
    \leq& \Prx_{\boldf}\big[\text{$S$ is a certificate for $\boldf$'s value on $x^\star$}\big] + \Prx_{\boldf}\big[\text{$\boldf$'s value is $1$ on at least one of $\mathcal{A}$'s queries}\big] \\
      \leq& \frac{k  \ell}{n} + \frac{q}{2^k}. \tag{\Cref{prop:all-zero-queries,prop:overlap}}
\end{align*}
\end{proof} 

}

% \gray{ Lastly, we complete the proof of \Cref{lem:non-monotone-lb}. Let $S$ be the set of coordinates output by $\mathcal{A}$ when every query it makes outputs $0$. Then,
% \begin{align*}
%     \Pr_{\bf}\big[\mathcal{A}& \text{ finds a correct certificate for }x^\star \big] \\
%     &= \Pr_{\bf}\big[\mathcal{A}\text{ finds a correct certificate for }x^\star , \text{ all queries output $0$}\big] \\
%     &\quad+\Pr_{\bf}\big[\mathcal{A} \text{ finds a correct certificate for }x^\star , \text{ some query outputs $1$}\big]\\
%     &\leq \Pr_{\bf}\big[S \text{ is a certificate for }x^\star\big] + \Pr_{\bf}\big[\text{Some query outputs $1$}\big] \\
%     &\leq \frac{k  \ell}{n} + \frac{q}{2^k}. \tag{\Cref{prop:all-zero-queries,prop:overlap}}
% \end{align*}
% \qed
% }

\subsection{Lower bound on the sample complexity of certifying a monotone function}

\begin{claim}
\label{claim:random-examples-lb}
For $k\le \ell \le cn$ where $c$ is a sufficiently small constant.  Suppose $\mathcal{A}$ is an algorithm which satisfies the following: given $q$ uniform random examples $(\bx,f(\bx))$ labeled by a monotone function $f : \zo^n\to\zo$ with $C(f)\le k$ and an input $x^\star \in \zo^n$, we have that $\mathcal{A}$ returns a size-$\ell$ certificate for $f$'s value on $x^\star$ w.h.p.  Then $q = \Omega(2^k)$. 
\end{claim} 

Combining~\Cref{claim:random-examples-lb} with the $q\ge \Omega(k\log n)$ we showed in~\Cref{lem:monotone-lb} yields the $q\ge \Omega(2^k + k\log n)$ lower bound stated in~\Cref{table}. 

\begin{proof} 
We will again apply Yao's lemma with $\boldf$ being a monotone conjunction of $k$ random variables and $\bx$ supported entirely on $x^\star = [1,\ldots,1]$.  (This is the same distribution as in the proof of~\Cref{lem:monotone-lb}.)  Let $\bQ$ be $q$ independent and uniform random elements $\bx^{(1)},\ldots,\bx^{(q)} \sim \zo^n$, and $\mathcal{A}$ be a deterministic algorithm.

By a union bound, 
\[ \Prx_{\bQ,\boldf}\big[\text{$\exists\,j\in [q]$ such that~$\boldf(\bx^{(j)})=1$} \big] \le \frac{q}{2^k}, \] 
and so if $q \le c2^k$ for a sufficiently small constant $c$, it then follows by Markov's inequality that: 
\begin{equation} \Prx_{\bQ}\left[\Prx_{\boldf}\big[\text{$\exists\,j\in [q]$ such that~$\boldf(\bx^{(j)})=1$}\big] \ge 0.01 \right] \le 0.01. \label{eq:Markov}
\end{equation} 
Fix a $Q = \{ x^{(1)},\ldots,x^{(q)}\} $ for which
\begin{equation}  \Prx_{\boldf}\big[\text{$\boldf(x^{(j)})=0$ for all $j\in [q]$}\big] \ge 0.99. \label{eq:good-Q}
\end{equation} 
Since $\mathcal{A}$ is deterministic, it has to return the same size-$\ell$ set, call it $S$, for all $f$'s that satisfy $f(x^{(i)}) = 0$ for all $j\in [q]$.  This set $S$ is a certificate for $\boldf$'s value on $x^\star = [1,\ldots,1]$ iff $\boldf$ is the conjunction of~$k$ variables $\bT$ where $\bT\sse S$, the probability of which is: 
\begin{equation}  \Prx_{\bf}[\bT\sse S] = \frac{{\ell\choose k}}{{n\choose k}}\le \frac{\left(\frac{e\ell}{k}\right)^k}{\left(\frac{n}{k}\right)^k} = \left(\frac{e\ell}{n}\right)^k \le 0.01,  \label{eq:unlikely-overlap} 
\end{equation} 
where the final inequality holds as long as $\ell \le cn$ for a sufficiently small constant $c$.  \Cref{eq:Markov,eq:good-Q,eq:unlikely-overlap} imply that $\mathcal{A}$ succeeds with probability at most $0.1$ over the randomness of $\boldf$, and the claim follows by Yao's lemma.  
\end{proof}

\section*{Acknowledgments}

We thank the STOC reviewers for their useful comments and feedback. 

Guy, Caleb, and Li-Yang are supported by NSF CAREER Award 1942123. Caleb is also supported by an NDSEG fellowship.  Jane is supported by NSF Award CCF-2006664.

\bibliography{bibtex}{}
\bibliographystyle{alpha}

\appendix
\section{\textit{p}-biased analysis}
\label{appendix:p_analysis}

\begin{proof}[Proof of \Cref{prop:basic_props}]
\begin{enumerate}
  \item We write
  \begin{align*}
      \Inf_p^{\oplus}[f]&=\sum_{i=1}^n\underset{\bx\sim\{0,1\}_p^n}{\Pr}[f(\bx)\neq f(\bx^{\oplus i})]\\
      &=\sum_{i=1}^n\underset{\bx\sim\{0,1\}_p^n}{\E}[\mathds{1}_{f(\bx)\neq f(\bx^{\oplus i})}]\\
      &=\underset{\bx\sim\{0,1\}_p^n}{\E}\left[\sum_{i=1}^n\mathds{1}_{f(\bx)\neq f(\bx^{\oplus i})}\right]\\
      &=\underset{\bx\sim\{0,1\}_p^n}{\E}[\Sens_f(\bx)].
  \end{align*}
  \item We write
  \begin{align*}
      \Inf_{i,p}^{\oplus}[f]&=\Pr_p[f(\bx)\neq f(\bx^{\oplus i})]\\
      &=\Pr_p[\bx_i=1]\Pr_p[f(\bx)\neq f(\bx^{\oplus i})\mid \bx_i=1]\\
      &\quad+\Pr_p[\bx_i=0]\Pr_p[f(\bx)\neq f(\bx^{\oplus i})\mid \bx_i=0]\\
      &=p\Pr_p[f_{x_i=1}(\bx)\neq f_{x_i=0}(\bx)]+(1-p)\Pr_p[f_{x_i=1}(\bx)\neq f_{x_i=0}(\bx)]\\
      &=\Pr_p[f_{x_i=1}(\bx)\neq f_{x_i=0}(\bx)].
  \end{align*}
 
  \item We write
  \begin{align*}
    \Inf_{i,p}^{\sim}[f]&=2\Pr_p[f(\bx)\neq f(\bx^{\sim i})]\\
    &=2\Bigl[(1-p)\Pr_p[x_i=1]\Pr_p[f(\bx)\neq f(\bx^{\oplus i})\mid \bx_i=1]\\
    &\quad +p\Pr_p[x_i=0]\Pr_p[f(\bx)\neq f(\bx^{\oplus i})\mid \bx_i=0]\Bigr]\\
    &=2\Bigl[(1-p)p\Pr_p[f_{x_i=1}(\bx)\neq f_{x_i=0}(\bx)]+p(1-p)\Pr_p[f_{x_i=1}(\bx)\neq f_{x_i=0}(\bx)]\Bigr]\\
    &=4p(1-p)\Pr_p[f_{x_i=1}(\bx)\neq f_{x_i=0}(\bx)]\\
    &=4p(1-p)\Inf_{i,p}^{\oplus}[f].
  \end{align*}
  
  \item Proof by induction on $n$. For $n=0$, $f$ is a constant in which case $\Var_p[f]=0\le \Inf_p^{\sim}[f]$. For the inductive step, let $i$ be any coordinate, then we first observe
  \begin{align}
      \Var_p[f]&=\E_p[f]-\E_p[f]^2 \nonumber\\
      &=\Bigl(p\E_p[f_{x_i=1}]+(1-p)\E_p[f_{x_i=0}]\Bigr)-\Bigl(p\E_p[f_{x_i=1}]+(1-p)\E_p[f_{x_i=0}]\Bigr)^2\nonumber\\
      &=p^2\Var_p[f_{x_i=1}]+(1-p)^2\Var_p[f_{x_i=0}]\nonumber \\
      &\quad +p(1-p)\Bigl[\Bigl(\E_p[f_{x_i=1}]-\E_p[f_{x_i=0}]\Bigr)^2+\Var_p[f_{x_i=1}]+\Var_p[f_{x_i=0}]\Bigr]\nonumber\\
      &=p\Var_p[f_{x_i=1}]+(1-p)\Var_p[f_{x_i=0}]+p(1-p)\Bigl(\E_p[f_{x_i=1}-f_{x_i=0}]\Bigr)^2 \label{eq:variance_decomposition}.
  \end{align}
%   \begin{align*}
%       \Var_p[f]&=\E_p[f]-\E_p[f]^2\\
%       &=\Bigl(p\E_p[f_{x_i=1}]+(1-p)\E_p[f_{x_i=0}]\Bigr)-\Bigl(p\E_p[f_{x_i=1}]+(1-p)\E_p[f_{x_i=0}]\Bigr)^2\\
%       &=p(1-p)\E_p[f_{x_i=1}]+p^2\Bigl(\E_p[f_{x_i=1}]-\E_p[f_{x_i=1}]^2\Bigr)\\
%       &\quad + p(1-p)\E_p[f_{x_i=1}]+(1-p)^2\Bigl(\E_p[f_{x_i=0}]-\E_p[f_{x_i=0}]^2\Bigr)\\
%       &\quad - 2p(1-p)\E_p[f_{x_i=1}]\E_p[f_{x_i=0}]\\
%       &=p^2\Var_p[f_{x_i=1}]+(1-p)^2\Var_p[f_{x_i=0}]\\
%       &\quad +p(1-p)\Bigl[\E_p[f_{x_i=1}]+\E_p[f_{x_i=0}]-2\E_p[f_{x_i=1}]\E_p[f_{x_i=0}]\Bigr]\\
%       &=p^2\Var_p[f_{x_i=1}]+(1-p)^2\Var_p[f_{x_i=0}]\\
%       &\quad +p(1-p)\Bigl[\Bigl(\E_p[f_{x_i=1}]-\E_p[f_{x_i=0}]\Bigr)^2+\Var_p[f_{x_i=1}]+\Var_p[f_{x_i=0}]\Bigr]\\
%       &=p\Var_p[f_{x_i=1}]+(1-p)\Var_p[f_{x_i=0}]+p(1-p)\Bigl(\E_p[f_{x_i=1}-f_{x_i=0}]\Bigr)^2.
%   \end{align*}
  Next, we observe
  \begin{align*}
      \E_p[f_{x_i=1}-f_{x_i=0}]&\le \E_p[|f_{x_i=1}-f_{x_i=0}|]\\
      &=\Pr_p[f_{x_i=1}(x)\neq f_{x_i=0}(x)]\\
      &=\Inf_{i,p}^{\oplus}[f]\le 2\Inf_{i,p}^{\oplus}[f]
  \end{align*}
  which shows that
  \begin{equation}
  \label{eq:expectation_difference_influence_lb}
  p(1-p)\Bigl(\E_p[f_{x_i=1}-f_{x_i=0}]\Bigr)^2\le 4p(1-p)\Inf_{i,p}^{\oplus}[f]=\Inf_{i,p}^{\sim}[f].
  \end{equation}
  Finally, we use the above in conjunction with the inductive hypothesis to get
  \begin{align*}
      \Var_p[f]&=p\Var_p[f_{x_i=1}]+(1-p)\Var_p[f_{x_i=0}]\\
      &\quad +p(1-p)(\E_p[f_{x_i=1}]-\E_p[f_{x_i=0}])^2 \tag{\Cref{eq:variance_decomposition}}\\
      &\le p\Var_p[f_{x_i=1}]+(1-p)\Var_p[f_{x_i=0}]+\Inf_{i,p}^{\sim}[f] \tag{\Cref{eq:expectation_difference_influence_lb}}\\
      &\le p\Inf_p^{\sim}[f_{x_i=1}]+(1-p)\Inf_p^{\sim}[f_{x_i=0}]+\Inf_{i,p}^{\sim}[f]  \tag{inductive hypothesis} \\
      &=\Inf_p^{\sim}[f].
  \end{align*}
  
  \item When $f$ is monotone we have $\Pr_p[f_{x_i=1}(\bx)\neq f_{x_i=0}(\bx)]=\E_p[f_{x_i=1}-f_{x_i=0}]=\E_p[f_{x_i=1}]-\E_p[f_{x_i=0}]$ and hence \Cref{prop:basic_props:infi_pr} can be rewritten as 
  $$
  \Inf_{i,p}^{\oplus}[f]=\E_p[f_{x_i=1}]-\E_p[f_{x_i=0}].
  $$
  Thus, we can write, using the law of total expectation
  \begin{align*}
      \E_p[f]&=\Pr_p[x_i=1]\E_p[f\mid x_i=1]+\Pr_p[x_i=0]\E_p[f\mid x_i=0]\\
      &=p\E_p[f_{x_i=1}]+(1-p)\E_p[f_{x_i=0}]\\
      &=p\left(\E_p[f_{x_i=1}]-\E_p[f_{x_i=0}]\right)+\E_p[f_{x_i=0}]\\
      &=\E_p[f_{x_i=0}]+p\Inf_{i,p}^{\oplus}[f]
  \end{align*}
  and analogously $p\E_p[f_{x_i=1}]+(1-p)\E_p[f_{x_i=0}]=\E_p[f_{x_i=1}]+(p-1)\E_p[f_{x_i=1}]-(p-1)\E_p[f_{x_i=0}]=\E_p[f_{x_i=1}]-(1-p)\Inf_{i,p}^{\oplus}[f]$.    
\end{enumerate}
\end{proof}

\section{Russo-Margulis Lemma}
\label{appendix:rm_lemma}
In this section we give a self-contained proof of the Russo-Margulis lemma. We adapt an exposition of the proof from \cite{G99} to the Boolean function setting here.

\begin{proof}[Proof of \Cref{lemma:russo_margulis}]
  The key step in our proof will be to generalize $\mathbb{E}_p[f]$ to a multivariate function. Specifically let $\bp=(p_1,\ldots,p_n)\in [0,1]^n$ and write $\{0,1\}^n_{\bp}$ to denote the distribution on $n$ bit strings where the $i$th bit is $1$ with probability $p_i$. Hence $\{0,1\}_p^n=\{0,1\}_{(p,\ldots,p)}^n$. We then define a function $H:[0,1]^n\to [0,1]$ by
  $$
  H(\bp)=H(p_1,\ldots,p_n)=\underset{\bx\sim\{0,1\}_{\bp}^n}{\Pr}{\left[f(\bx)=1\right]}.
  $$
  Note that $H(p,\ldots,p)=\mathbb{E}_p[f]$ and so it is sufficient to show the result holds for the derivative of $H$ with respect to $\bp$ evaluated at $(p,\ldots,p)$. The partial derivative of $H$ with respect to its $i$th input is given by
  $$
  \frac{\partial H}{\partial p_i}(\bp)=\lim_{\varepsilon\to 0}\frac{H(\bp+\boldsymbol{\varepsilon})-H(\bp)}{\varepsilon}
  $$
  where $\boldsymbol{\varepsilon}\in [0,1]^n$ is the vector with $\varepsilon$ in the $i$th entry and $0$s in all other entries. Let $T\subseteq \{0,1\}^n$ be the set of inputs $x$ for which $f(x)=1$. Furthermore, we partition $T$ into sets $T_0, T_1$ depending on whether $x_i$ is $0$ or $1$ for $x\in T$. Formally for $b\in\{0,1\}$ define
  $$
  T_b=\{x\in T:x_i=b\}.
  $$
  We write $\Pr_{\bp}[x]$ to denote the probability of $x$ under the distribution $\{0,1\}^n_{\bp}$ so that
  \begin{align*}
      H(\bp)&=\sum_{x\in T}\Pr_{\bp}[x]\\
      &=\sum_{x\in T_1}p_i\Pr_{\bp}[x \mid x_i=1]+\sum_{x\in T_0}(1-p_i)\Pr_{\bp}[x\mid x_i=0].
  \end{align*}
  Since $\Pr_{\bp}[x \mid x_i=1]$ doesn't depend on the $i$th entry of $\bp$ we have that $\Pr_{\bp}[x \mid x_i=1]=\Pr_{\bp+\boldsymbol{\varepsilon}}[x \mid x_i=1]$  and likewise for $\Pr_{\bp}[x \mid x_i=0]$. Thus we have
  \begin{align*}
      H(\bp+\boldsymbol{\varepsilon})-H(\bp)&=\sum_{x\in T_1}(\varepsilon+p_i-p_i)\Pr_{\bp}[x \mid x_i=1]+\sum_{x\in T_0}((1-p_i-\varepsilon) -(1-p_i))\Pr_{\bp}[x\mid x_i=0]\\
      &=\sum_{x\in T_1}\varepsilon\Pr_{\bp}[x \mid x_i=1]-\sum_{x\in T_0}\varepsilon\Pr_{\bp}[x\mid x_i=0].
  \end{align*}
  If $f$ is monotone and $f(x)=1$ for some $x\in T_0$ then $f(x^{\oplus i})=1$ and so $x^{\oplus i}\in T_1$. Moreover, since $x$ and $x^{\oplus i}$ differ only on the $i$th bit, $\Pr_{\bp}[x \mid x_i=0]=\Pr_{\bp}[x^{\oplus i} \mid x^{\oplus i}_i=1]$ and so the two terms cancel each other in the above summation. The only terms left will be those $x\in T_1$ with no counterpart in $T_0$ which means $f(x)=1$ but $f(x^{\oplus i})=0$. It follows that
  \begin{align*}
      H(\bp+\boldsymbol{\varepsilon})-H(\bp)&=\sum_{x\in T_1}\varepsilon\Pr_{\bp}[x \mid x_i=1]-\sum_{x\in T_0}\varepsilon\Pr_{\bp}[x\mid x_i=0]\\
      &=\varepsilon\sum_{\substack{x\in T_1\\ f(x)\neq f(x^{\oplus i})}}\Pr_{\bp}[x \mid x_i=1]\\
      &=\varepsilon\underset{\bx\sim\{0,1\}_{\bp}^n}{\Pr}[f(\bx)\neq f(\bx^{\oplus i})]
  \end{align*}
  where the last equality follows from the observation that if $x\in T_1$ and $f(x)\neq f(x^{\oplus i})$ then
  \begin{align*}
      \Pr_{\bp}[x\mid x_i=1]&=p_i\Pr_{\bp}[x\mid x_i=1]+(1-p_i)\Pr_{\bp}[x^{\oplus i}\mid x^{\oplus i}_i=0]\\
      &=\Pr_{\bp}[x]+\Pr_{\bp}[x^{\oplus i}].
  \end{align*}
  We can now write the partial derivative of $H$ with respect to its $i$th input as
  $$
  \frac{\partial H}{\partial p_i}(\bp)=\lim_{\varepsilon\to 0}\frac{\varepsilon \underset{\bx\sim\{0,1\}_{\bp}^n}{\Pr}[f(\bx)\neq f(\bx^{\oplus i})]}{\varepsilon}=\underset{\bx\sim\{0,1\}_{\bp}^n}{\Pr}[f(\bx)\neq f(\bx^{\oplus i})].
  $$
  Now using the multivariate chain rule and evaluating at $\bp=(p,\ldots,p)$ we compute
  \begin{align*}
      \frac{d}{dp}\mathbb{E}_p[f]&=\left.\frac{dH}{d\bp}(\bp)\right\rvert_{\bp=(p,\ldots,p)}\\
      &=\left.\sum_{i=1}^n\frac{\partial H}{\partial p_i}(\bp)\right\rvert_{\bp=(p,\ldots,p)}\\
      &=\left.\sum_{i=1}^n\underset{\bx\sim\{0,1\}_{\bp}^n}{\Pr}[f(\bx)\neq f(\bx^{\oplus i})]\right\rvert_{\bp=(p,\ldots,p)}\\
      &=\sum_{i=1}^n\Inf_{i,p}^{\oplus}[f]\\
      &=\Inf_p^{\oplus}[f].
  \end{align*}
\end{proof}

\section{Angluin's Algorithm}
\label{appendix:angluin}

In this section we give an overview of Angluin's algorithm adapted to our setting and a proof of correctness.

\begin{algorithm}
    \begin{algorithmic}[1]
    \Require A monotone function $f:\{0,1\}^n\to \{0,1\}$, a $b$-certificate $S$ for $b\in\{0,1\}$, and input $x^\star$.
    \State $\texttt{SEEN}\leftarrow \varnothing$
    \State Initialize $z_S\in\{0,1\}^n$ to be equal to $x^\star$ on coordinates in $S$ and $1-b$ everywhere else
    \While{$|S|\le C(f)$}
    \State Pick some $i\in S\setminus\texttt{SEEN}$
    \State If $f(z_S^{\oplus i})\neq f(z_S)$ then add $i$ to \texttt{SEEN}, otherwise remove $i$ from $S$ and update $z_S$
    \EndWhile
    \State \Return $S$.
    \end{algorithmic}
    \caption{Reducing a certificate}
    \label{alg:trim_certificate}
\end{algorithm}

\begin{proof}[Proof of \Cref{claim:trim_cert}]
    \Cref{alg:trim_certificate} gives a sketch of the procedure. Suppose without loss of generality that $f(x^\star)=1$ and so $S$ is a $1$-certificate. We can continuously attempt to remove coordinates from $S$ one at a time until $|S|\le C(f)$. For a 1-certificate $S$, write $z_S\in\{0,1\}^n$ for the string which has a $1$ at each coordinate in $S$ and $0$s everywhere else. Note that $z_S\le x^\star$, $f(z_S)=1$, and also $z_S\le y$ for all $y$ satisfying $y\rvert_{S}=z_S\rvert_{S}$. For $i\in S$, we check if $i$ is an irrelevant coordinate (\Cref{def:irrelevance}) by checking if flipping the $i^{\text{th}}$ coordinate in $z_S$ flips the output of the function. That is, we check if $i$ is sensitive on $z_S$. If $i$ is not sensitive, we remove $i$ from $S$ and recurse on $S\setminus\{i\}$. Otherwise, we leave $i$ in $S$ and do not check it again. We proceed in this fashion until $|S|\le C(f)$. Since we only check coordinates in $S$ and check each such coordinate at most once we make $\le 2|S|$ queries to $f$.

    To establish correctness, suppose this procedure returns $S'$. Since we only remove non-sensitive coordinates from $S$ we have $f(z_{S'})=1$. For any $y$ satisfying $y\rvert_{S'}=z_{S'}\rvert_{S'}$ we know that $y\ge z_{S'}$ and hence $f(y)=1$ by monotonicity. It follows that $S'$ is a $1$-certificate for $z_{S'}$ and likewise for $x^\star$ as $S'\subseteq S$. Note also that if $i$ is in $S$ and $i$ is sensitive for $z_S$ then $i$ remains sensitive for all $z_{S'}$ with $i\in S'\subseteq S$. In particular, $z_{S'}\le z_S$ and $z_{S'}^{\oplus i}\le z_S^{\oplus i}$ which shows $0=f(z_S^{\oplus i})\ge f(z_{S'}^{\oplus i})$ by monotonicity. Thus, any sensitive coordinate can be left in the certificate without having to check again. Moreover, since $\Sens_f(z_S)\le C(f)$ we know that the number of sensitive indices we keep in the certificate $S$ is at most $C(f)$ which ensures that if $|S|>C(f)$ there will always be some non-sensitive index that we can remove from $S$.
\end{proof}

\end{document}